\documentclass[11pt,letterpaper]{article}
\usepackage[english]{babel}
\usepackage[utf8]{inputenc}
\usepackage{amsmath}
\usepackage{amsthm, amssymb}
\usepackage{amsfonts}
\usepackage{mathtools,thmtools}
\usepackage{algorithm}
\usepackage{algorithmic}
\usepackage{graphicx}
\usepackage{tikz}
\usetikzlibrary{automata,arrows,positioning,calc}
\usepackage{fullpage}
\usepackage{bbm}
\usepackage[
  pagebackref,
  colorlinks=true,
  urlcolor=blue,
  linkcolor=blue,
  citecolor=blue,
]{hyperref}
\usepackage[nameinlink]{cleveref}
\usepackage{tikz}
\usetikzlibrary{automata,arrows,positioning,calc}

\newcommand{\OLH}{1-\textsc{Lookahead}}
\newcommand{\cred}{\textsc{Cred}}

\newtheorem{theorem}{Theorem}[section]

\newtheorem{lemma}[theorem]{Lemma}
\newtheorem{proposition}[theorem]{Proposition}
\newtheorem{claim}[theorem]{Claim}
\newtheorem{definition}[theorem]{Definition}

\newtheorem*{conjecture*}{Conjecture}
\newtheoremstyle{nonindented}{1ex}{1ex}{}{}{\bfseries}{.}{.5em}{}
\newtheoremstyle{indented}{1ex}{1ex}{\itshape\addtolength{\leftskip}{0.6cm}\addtolength{\rightskip}{0.6cm}}{}{\bfseries}{.}{.5em}{}
\theoremstyle{nonindented}

\theoremstyle{indented}

\newtheorem*{direction*}{Research Direction}
\theoremstyle{plain}

\newenvironment{alignedequation}{\begin{equation} \begin{aligned}}{\end{aligned} \end{equation}}
\newenvironment{alignedequation*}{\begin{equation*} \begin{aligned}}{\end{aligned} \end{equation*}}

\newenvironment{proofsketch}{\begin{proof}[Proof Sketch]}{\end{proof}}

\DeclareMathOperator{\sk}{\mathsf{sk}}
\DeclareMathOperator{\pk}{\mathsf{pk}}

\newcommand{\lt}{\left}
\newcommand{\rt}{\right}

\def\min{\qopname\relax n{min}}
\def\max{\qopname\relax n{max}}

\def\argmin{\qopname\relax n{argmin}}

\def\Pr{\qopname\relax n{\mathbf{Pr}}}

\newcommand{\NN}{\mathbb{N}}

\def\X{\mathcal{X}}

\def\Z{\mathcal{Z}}

\def\eps{\epsilon}

\newcommand{\Exp}{\mathsf{Exp}}

\newenvironment{lp*}{\begin{equation*}  \begin{array}{lll}}{\end{array}\end{equation*}}

\DeclarePairedDelimiter{\sq}{[}{]}
\DeclarePairedDelimiter{\paren}{\lparen}{\rparen}
\DeclarePairedDelimiterX{\cond}[2]{[}{]}{#1\,\delimsize\vert\,\mathopen{} #2}
\newcommand{\mattnote}[1]{\textcolor{blue}{#1}}

\begin{document}
% Title portion. Note the short title for running heads 
\title{Profitable Manipulations of Cryptographic Self-Selection are Statistically Detectable}
\author{Linda Cai
\thanks{Princeton University, \tt{tcai@princeton.edu}.} \and Jingyi Liu \thanks{Princeton University, \tt{jingyi.liu@princeton.edu}.} \and S. Matthew Weinberg 
\thanks{Princeton University, \tt{smweinberg@princeton.edu}.} \and Chenghan Zhou \thanks{Stanford University, \tt{chzhou@stanford.edu}.}} 
\date{}

% note that the abstract must come before \maketitle
\maketitle

\begin{abstract}
Cryptographic Self-Selection is a common primitive underlying leader-selection for Proof-of-Stake blockchain protocols. The concept was first popularized in Algorand~\cite{ChenM19}, who also observed that the protocol might be manipulable.~\cite{FerreiraHWY22} provide a concrete manipulation that is strictly profitable for a staker of any size (and also prove upper bounds on the gains from manipulation).

Separately,~\cite{YaishSZ23, BahraniW23} initiate the study of \emph{undetectable} profitable manipulations of consensus protocols with a focus on the seminal Selfish Mining strategy~\cite{EyalS14} for Bitcoin's Proof-of-Work longest-chain protocol. They design a Selfish Mining variant that, for sufficiently large miners, is strictly profitable yet also indistinguishable to an onlooker from routine latency (that is, a sufficiently large profit-maximizing miner could use their strategy to strictly profit over being honest in a way that still appears to the rest of the network as though everyone is honest but experiencing mildly higher latency. This avoids any risk of negatively impacting the value of the underlying cryptocurrency due to attack detection).

We investigate the detectability of profitable manipulations of the canonical cryptographic self-selection leader selection protocol introduced in~\cite{ChenM19} and studied in~\cite{FerreiraHWY22}, and establish that for any player with $\alpha < \frac{3-\sqrt{5}}{2} \approx 0.38$ fraction of the total stake, \emph{every strictly profitable manipulation is statistically detectable}. Specifically, we consider an onlooker who sees only the random seed of each round (and does not need to see any other broadcasts by any other players). We show that the distribution of the sequence of random seeds when any player is profitably manipulating the protocol is inconsistent with any distribution that could arise by honest stakers being offline or timing out (for a natural stylized model of honest timeouts). 
\end{abstract}
\newpage

\section{Introduction}
Since Nakamoto introduced Bitcoin in 2008, blockchain technology has made a significant impact on digital transactions by establishing a decentralized system in which transactions are validated through consensus among peers, rather than by a central authority. This innovation, while popularizing decentralized currencies, has also brought to light substantial challenges, particularly the extensive computational and energy demands of its proof-of-work (PoW) consensus mechanism. Notably, the energy consumption associated with Bitcoin mining exceeds that of many countries, raising significant environmental concerns. Furthermore, the necessity for large-scale mining hardware introduces considerable centralization risks to cryptocurrencies \cite{ArnostiW19}, many of which are inherently designed to be decentralized.

In response to these challenges, the blockchain community has been exploring Proof of Stake (PoS), which has been implemented in many prominent crypto-currencies (e.g. Ethereum, Algorand, Cardano). In each round, PoS selects block leaders (who get to propose a block to be included) based on the stake, reducing energy usage and aiming to prevent Sybil attacks by randomly assigning leadership chances proportionally to coin holdings. However, the leader selection process in PoS presents additional challenges. For example, the pseudorandomness resulting from PoW is in some sense ``external'' to the blockchain (the next miner is selected proportionally to their computational power, independently, and nothing in the blockchain itself can influence this). Replicating this property in PoS blockchains has proved challenging without trusting an external randomness beacon (which is often a non-starter in blockchain applications, whose entire purpose is to remove the need for such trust).\footnote{To slightly elaborate on this point: trusting a centralized external randomness beacon (such as NIST) is certainly a non-starter, because NIST then has control over the block producers. One could instead have an external distributed process to generate random numbers independent of this blockchain. But if this blockchain has monetary value, then securely implementing that distributed process is its own challenge. The story is getting more subtle with Verifiable Delay Functions that might act as a cryptographic external randomness beacon, although their security assumptions are hardware-based and not as battle-tested as standard cryptography, so there will always be a desire for solutions based on standard cryptography.} On the other hand, pseudorandom numbers generated using the blockchain itself can often be predicted by the miners, opening up the possibility of profitable deviations \cite{BrownCohenNPW19}.

One promising idea in addressing the leader selection challenge is cryptographic self-selection, initially proposed by Algorand~\cite{ChenM19}. Cryptographic self-selection is a protocol to select a block-proposer for round $r+1$ as a function of communication during round $r$. We overview Algorand's canonical proposal shortly, and briefly note here that it is known to admit profitable deviations for arbitrarily small participants~\cite{FerreiraHWY22}.\footnote{This is in contrast to block-witholding manipulations in PoW longest-chain protocols~\cite{SapirshteinSZ16, KiayiasKKT16, EyalS14}, although alternate strategic manipulations of some PoW protocols are profitable for arbitrarily small miners~\cite{FiatKKP19, GorenS19, YaishSZ23}. } We subsequently discuss cryptographic self-selection in further detail, but at this point merely wish to note that: (a) nonmanipulable randomness sources are a major open problem within the blockchain community, due to applications for PoS, (b) these problems are important to both researchers~\cite{BrownCohenNPW19,ChenM19, FerreiraW21,FerreiraHWY22} and practitioners ~\cite{RanDAO2020} \footnote{For example, this blog post by the Ethereum foundation on manipulating its RanDAO: \href{https://ethresear.ch/t/selfish-mixing-and-randao-manipulation/16081}{link}.}, and (c) the particular approach initially proposed in~\cite{ChenM19} is a canonical testbed due to its elegance and simplicity (which we overview shortly).

Separately, recent work of~\cite{YaishSZ23, BahraniW23} propose a novel concern for profitable manipulations: detectability. Specifically, while it may be challenging to \emph{trace} a strategic manipulation to a particular actor in a permissionless system,\footnote{This is not to say that tracing strategic manipulations is impossible -- indeed, law enforcement regularly traces attacks in permissionless systems: \href{https://www.chainalysis.com/law-enforcement/}{link}.} profitable manipulations would likely be \emph{detectable}. This observation serves as a basis to mitigate concerns with profitable manipulations in practice -- perhaps the manipulator will earn moderate additional cryptocurrency via manipulation, but its detection may cause the value of these tokens to tank when measured in USD. Their work highlights that detectability of strategic manipulations plays a significant role in their usability in practice -- undetectable deviations avoid the risk of devaluing the underlying cryptocurrency, while detectable ones can be disincentivized through outside-the-model means.\\

Our paper lies at the intersection of these two agendas: we investigate the detectability of profitable manipulations in cryptographic self-selection. Surprisingly, our main result finds that for any participant with less than $\frac{3-\sqrt{5}}{2}\approx 0.38$ fraction of the total stake, \emph{all profitable manipulations of Algorand's canonical cryptographic self-selection protocol are statistically detectable}.\\

We now provide additional context and details for our result.

\vspace{8pt} \emph{Leader Selection in PoS Blockchains.} PoS consensus protocols typically take one of two forms: they may be a \emph{longest-chain protocol}, or a \emph{Byzantine Fault Tolerant (BFT)-based protocol}. Both formats are well-represented in practice, and Ethereum is in some sense a hybrid of the two. In a longest-chain protocol, strategic manipulations are more straight-forward -- they typically come by inducing forks, and causing the attacker to have their own blocks represent a greater fraction of blocks in the longest chain~\cite{EyalS14, SapirshteinSZ16, KiayiasKKT16, BrownCohenNPW19, FerreiraW21}. Manipulations of BFT-based protocols are more subtle. BFT-based protocols proceed one block at a time, reaching a strong consensus on block $r$ and finalizing it forever before getting to work on block $r+1$. As such, these protocols are nonmanipulable at the per-block level (unless the attacker has sufficient stake to cause significantly more damage by violating consensus entirely). Instead, these protocols typically have a randomly-selected ``leader'' dictate the contents of the block and the per-round BFT protocol aims to reach consensus on the leader's block. But, these protocols still need an effective method to select a leader for each round independently and proportional to their stake. 

Fortunately for mechanism designers, leader selection protocols are often modular components of the broader blockchain protocol, and can be studied in isolation from the (significantly more complex) BFT protocols that handle per-round consensus. 

\vspace{8pt} \emph{Algorand's Canonical Leader Cryptographic Self-Selection.}~\cite{ChenM19} propose an elegant leader selection protocol, which we describe for simplicity in the case where each account holds the same number of coins (we rigorously overview their protocol in the general case in Section~\ref{sec:model}, but omit the generalization now in the interest of clarity). First, pick a uniformly random seed, $Q_1$, for round one. Then in round $r$, ask each account holder $i$ to first digitally sign $Q_r$ and then hash\footnote{The formal concept is a Verifiable Random Function, which we define in Section~\ref{subsec:consensusAndScore}. Intuitively, the hash is a uniformly random number drawn specifically for player $i$ in a manner that no other player can precompute (because they can't digitally sign on behalf of player $i$).} their digital signature to get a credential $\cred_i^r$. Whoever broadcasts the smallest credential is the leader for round $r$. 

Their protocol has several desirable properties. First, assuming that every player honestly digitally signs and hashes in each round (and that the hash function behaves like a random oracle), the leader in each round is indeed a uniformly random coin, independent of all previous rounds. Second, it is not predictable too far into the future: because player $i$ cannot digitally sign on behalf of player $j$, player $i$ has absolutely no idea what seed might result next round (if another player is the leader). Finally, the manner in which it can possibly be manipulated is extremely structured: the only strategies available to a player are to broadcast or not broadcast their credentials.

Still, their protocol is not perfect --~\cite{ChenM19} already acknowledge that it might be manipulable, and~\cite{FerreiraHWY22} establish a strictly profitable strategy for arbitrarily small players.~\cite{FerreiraHWY22}'s strategy is fairly simple, and we overview it in Section~\ref{sec:example}. 

\vspace{8pt} \emph{Detecting Strategic Behavior in Cryptographic Self-Selection.} How would one detect that a participant is strategically manipulating a protocol? In PoW longest-chain protocols,~\cite{BahraniW23} propose to look at the pattern of forks -- strategic behavior often results in long runs of consecutive forks whereas routine latency instead would result in independently distributed forks. This particular detection method is not applicable to a BFT-based protocol, as BFT-based protocols have no forks once a block is finalized.

In the spirit of~\cite{BahraniW23}, we aim to detect strategies using the minimal amount of information possible, and in particular we only use information that is available to anyone following the blockchain. Specifically, anyone following the blockchain \emph{must know who is the leader of round $r$} and \emph{must know their credential $\cred_i^r$} that proves they are the leader. 

If every participant in the network were honest, and there were $n$ coins in the network, we would expect in a given round that the winning credential is distributed according to the minimum of $n$ independent draws of the Hash function. So across a large number of rounds, an observer could check the sequence of winning credentials and see if they empirically match i.i.d.~draws of the minimum of $n$ independent draws of the Hash function.

This is perhaps too strong of an assumption on honest parties, however. In particular, it assumes either that every single coin is online and participating in the protocol or that the observer otherwise knows that exactly (say) $k$ coins are online. An observer instead might know that there exists some number $k$ of online coins participating in the protocol, but not know $k$. Then, they would expect to see sequences of credentials that empirically match i.i.d.~draws from the minimum of $k$ independent draws of the Hash function, for some $k$. 

So consider an attacker who controls multiple accounts. They can selectively refrain from broadcasting in round $r$ (and might benefit from doing so, if another account will win round $r$ anyway and their chosen credential gives them a better shot of winning round $r+1$), but doing so will skew the distribution of round $r$'s credential larger and the distribution of round $r+1$'s credential smaller. This is profitable, but when done naively detectable (we analyze~\cite{FerreiraHWY22}'s particularly simple strategy, which follows from this intuition, in Section~\ref{sec:stat-detec-meth}). The challenge for the attacker is whether it is possible to profit (by biasing their winning credential to be lower in some rounds), without being detectable (by biasing the winning credential in other rounds to be higher). Our main result shows that this is impossible: for any $\alpha < \frac{3-\sqrt{5}}{2} \approx 0.38$,\footnote{Note, for example, that an adversary with $\alpha > \frac{3-\sqrt{5}}{2} > 1/3$ of the stake could alternatively directly violate the underlying consensus protocol, which would do significantly more damage than a strategic manipulation.} and any participant with an $\alpha$ fraction of the total stake, any strategy that leads a $>\alpha$ fraction of the rounds produces a distribution over sequences of winning credentials that is not consistent with any number of online honest coins.\footnote{Like most prior work (e.g.~\cite{KiayiasKKT16, BrownCohenNPW19, FerreiraW21, BahraniW23}), we consider an attacker who does not have excessively strong network connectivity -- see Section~\ref{sec:stat-detec-meth} for the formal setup.} 

Finally, one might even consider having a fixed $k$ of online coins to be too stringent of a null hypothesis -- perhaps the number of active coins fluctuates from round to round. We also establish that our main result degrades smoothly in the deviation an observer is comfortable attributing to honest-but-occasionally-offline behavior. If the observer believes the online coins to fluctuate within $1\pm \delta$ of an unknown baseline, and the true online coins indeed fluctuate within $1\pm \delta$ of some ground truth baseline, undetectable manipulations lead at most an additional $2\delta$ fraction of rounds.

\vspace{8pt} \emph{Roadmap and Discussion.} We study the detectability of strategic manipulations in cryptographic self-selection, Algorand's canonical leader selection protocol~\cite{ChenM19}. We establish that any profitable deviation is detectable, and also quantitatively extend our results to even further relaxed null hypotheses of what might result from honest-but-offline behavior. 

Detectability of profitable manipulations is a desirable property of consensus protocols, as it provides an outside-the-model avenue to deter deviant behavior. While~\cite{BahraniW23} derive profitable, undetectable deviations from longest-chain PoW consensus protocols, we instead show that cryptographic self-selection admits no profitable manipulations. Our work now establishes that some canonical protocols admit undetectable profitable deviations while others do not, and further motivates detectability of profitable manipulations as a standard question to be asked of novel consensus protocols.

In Section~\ref{sec:intro:related-work}, we overview related work in further detail. In Section~\ref{sec:model} and \Cref{sec:stat-detec-meth} we overview our model and our statistical detection methods in significantly more detail. Section~\ref{sec:example} overviews the profitable strategy of~\cite{FerreiraHWY22} through the lens of detectability in order to familiarize the reader with the techniques. Section~\ref{sec:detect} formally states and proves our main result and its robust extension.

\subsection{Related Work} \label{sec:intro:related-work}

\vspace{8pt}

\emph{Detection of Strategic Attacks in Proof-of-Work Protocols.}
Several methods of detecting selfish mining in proof-of-work protocol have been proposed. \cite{ChicarinoAJR20} presents a heuristic to detect selfish mining based on changes in the height of forks in a blockchain network and their simulation result implies a connection between the presence of selfish mining attack and higher rate of forks, with a mean height of higher than 2. 

\cite{LiCT22} proposes a statistical test for each miner based on the null hypothesis that under honest mining, the probability of observing two successive blocks mined by the same miner is given by type II binomial distribution of order 2, and the presence of selfish mining will cause deviation from such distribution, causing a higher probability of observing successive blocks mined by the same miner. The authors conduct empirical tests on five cryptocurrencies based on Proof-of-Work---Bitcoin, Litecoin, Ethereum, Monacoin and Bitcoin Cash and claim to be the first research work that reveals the presence of selfish mining in real cryptocurrency systems, although they acknowledge that other reasons can also lead to abnormal successive block discovery rates. We further note that their detection method relies on knowing which addresses or wallets are controlled by the same user, while our detection scheme does not rely on such knowledge.

Other works such as \cite{WangLLWY21} use neural networks that achieve good accuracy of detecting selfish mining on simulated datasets.

In contrast to these detection methods, \cite{BahraniW23} proves the existence of a statistically undetectable and strictly profitable selfish mining strategy for miners with 38.2\% of the total hash rate.  Under this strategy, the attacker hides their block with carefully constructed probabilities such that the eventual structure of the blockchain under this selfish mining attack has the same distribution as the structure of the blockchain constructed by only honest miners with a different latency parameter. Thus statistical tests that only look at the pattern of the blockchain itself such as fork heights cannot detect their attack.

\vspace{8pt} \emph{Strategic Manipulation of Consensus Protocols.} Following seminal work of~\cite{EyalS14}, there is now a long body of work studying strategic manipulations in consensus protocols~\cite{EyalS14, SapirshteinSZ16, KiayiasKKT16, CarlstenKWN16, GorenS19, FiatKKP19, FerreiraW21, FerreiraHWY22, YaishTZ22, YaishSZ23, BahraniW23}. These works are all thematically related to ours in that we also study strategic manipulation of consensus protocols. Of these, only~\cite{FerreiraHWY22} bears any technical similarities, as the others all study longest-chain variants. 

In terms of motivating cryptographic self-selection,~\cite{BrownCohenNPW19} establish that longest-chain variants with fully-internal pseudorandomness are all vulnerable to a selfish-mining-style attack based on predicting future randomness. 

Research works on the detection of strategic attacks mostly focus on the longest-chain Proof-of-Stake protocols such as~\cite{NeuderMRP20}. To the best of our knowledge, our work is the first to propose a detection method for manipulating leader selection protocols in BFT-based blockchains.

\vspace{8pt} \emph{Relevant Proof-of-Stake Protocols in Practice.} Several large blockchains employ Proof-of-Stake over Proof-of-Work, and there is not yet convergence on a dominant consensus paradigm. For example, Cardano~\cite{KiayiasRDO17} uses a longest-chain variant considered in~\cite{BrownCohenNPW19}, Algorand uses cryptographic self-selection~\cite{GiladHMVZ17, ChenM19} considered in~\cite{FerreiraHWY22} (although Algorand seems to have since updated their leader selection to induce a round robin aspect -- every $k$ rounds, the winner's credential sets the seeds for the subsequent $k$ rounds. See \cite{GiladHMVZ17}.), and Ethereum uses a hybrid of the two (although manipulations of Ethereum are much closer to manipulations of cryptographic self-selection than of longest-chain protocols -- see \href{https://ethresear.ch/t/selfish-mixing-and-randao-manipulation/16081}{here}). In terms of relevance for practice, our results (a) highlight a desirable property of~\cite{ChenM19}'s original cryptographic self-selection that is desirable in practice, and (b) serve as a canonical example to highlight manners in which a protocol might avoid undetectable profitable deviations.

\section{Model and Preliminaries} \label{sec:model}

\subsection{Proof-of-Stake Consensus Protocols with Finality} \label{subsec:consensusAndScore}
Proof-of-stake protocols with finality look more like classical Consensus algorithms from Distributed Systems than Bitcoin's Longest-Chain protocol. That is, these protocols repeatedly run a secure consensus algorithm to agree on a block of authorized transactions, add this block of transactions to the ledger, and proceed. Unlike the Longest-Chain protocol, these blocks are added to the ledger and remain in the ledger forever. In order to maintain security guarantees, the consensus algorithm for each block is often complex.

To mitigate this complexity (both computational, communication, and conceptual), many protocols select a \emph{leader} $\ell_t$ who plays a special role in the consensus protocol. Intuitively, all participants try to copy the leader's proposed block. Similarly to a Longest-Chain protocol, the leader $\ell_t$ dictates the contents of the block. That is, the contents of Block $t$ are fully dictated by the leader $\ell_t$, just like in a Longest-Chain protocol (and the only difference is how consensus is reached so that the rest of the network agrees on what block was indeed dictated). As a result, we model the payoff of players to be the probability of being a block leader, since Thus, being the block leader is the only way where players could gain profits (either from block rewards or MEV)

%They typically run a consensus protocol to select a leader $\ell_t$  to propose a new block in each round $t$, and maintain a ledger of blocks $B_0$, $B_1$, ..., $B_t$ created in each round. Since the leader is uniquely determined in each round, the ledger is simply a linear sequence of blocks and has no forks.

% The proof-of-stake protocols with finality  maintain a ledger of blocks, $B_0$, $B_1$, ..., $B_t$, with the key distinction that there are no forks; once a block for round $t$ is established, it is permanent. A leader, denoted as $l_t$, is selected in each round to propose the new block $B_t$, based on the sequence of previous blocks and is responsible for proposing the new block $B_t$.

% 2.2 Cryptographic Self-Selection for Leader Determination
% Determining the leader $l_t$ in a blockchain with finality is critical. To mitigate predictability and grinding attacks, the leader-selection protocol must ensure:

% The distribution of each `t is proportional to stake and independent across rounds.
% Users have limited ability to predict or influence their chances of becoming the leader.
In order to mitigate grinding attacks, the leader-selection protocol typically needs to ensure that when participants are behaving honestly, the probability that each participant (or pool of participants) gets to be the leader in each round is proportional to each participant's stake (hence the name ``proof-of-stake"). Moreover, there should be limited room for a participant to gain extra profit by deviating from the protocol. 

Cryptographic self-selection (used by Algorand \cite{ChenM19, GiladHMVZ17}) is an elegant solution for the leader-selection protocol, which does not rely on the existence of frequent and high quality randomness beacon to generate a random seed for each round. Instead, it uses information about the previous rounds to generate a seed for the current round. We now briefly describe the protocol and the parts that are relevant for constructing a statistical detection method of strategic deviation from the protocol. The reader could refer to \cite{FerreiraHWY22, ChenM19} for a more detailed description of the protocol and encryption schemes.  

% The core components are an Ideal Verifiable Random Function (VRF) and a balanced scoring function.

% What I want to say next: the process, minus cryptographic detail, can be sylized as follows: 1) scoring function for each account with sake; 2) the bidder gets to choose their account 

% if we want to add more detail on cryptography later, we can do it without impacting notations afterwards 

% Essential Cryptographic Tools. An Ideal VRF, essential for cryptographic self-selection, has these properties:

There are two key components to the cryptographic self-selection protocol: Verifiable Random Functions (VRFs) and balanced scoring functions. 

% \lindac{Change Eval to $f_{sk}$ according to our definition. }
\begin{definition}[Verifiable Random Function (VRF) \cite{MicaliSR99}] \label{def:vrf}
A Verifiable Random Function is a public-key cryptographic function that generates public key and secret key pairs, denoted $(\sk,\pk)$, and efficiently evaluates an input $x$ using a function $f_{\sk}$ that is dependent on the secret key. The function produces an output $y$ and a proof of correctness, which can be verified efficiently by anyone who has the public key. The following security properties are guaranteed:
% A Verifiable Random Function is a tuple of three polynomial-time algorithms $(\Gen, \Eval, \Verify)$ satisfying the following properties:
% \begin{enumerate}
% \item 
% Key Generation $(\Gen)$: A probabilistic algorithm that generates a pair of keys:
%    \[
%    (\sk, \pk) \leftarrow \Gen(1^\lambda)
%    \]
%    where $\lambda$ is the security parameter.
% \item 
% Evaluation $(\Eval)$: A deterministic algorithm that, given a secret key $\sk$ and an input $x$, produces a pseudorandom output $y$ and a proof of correctness:
%    \[
%    (y, proof) \leftarrow \Eval(\sk, x)
%    \]
%    % \lindac{Change $\pi$ to sth else. }
% \item 
% \lindac{Possibly change to probabilistic definition.}
% Verification $(\Verify)$: A deterministic algorithm that, given a public key $\pk$, an input $x$, an output $y$, and a proof, verifies the correctness of the output: 
%    \[
%    \{0, 1\} \leftarrow \Verify(\pk, x, y, proof)
%    \]
%    where it outputs 1 if the proof is valid and 0 otherwise.
% \end{enumerate}
% \textbf{Security Properties:}
\begin{itemize}
\item 
Pseudorandomness: Given the public key $\pk$ and a sequence of input-output pairs $(x_1, y_1),\ldots,$ $(x_n, y_n)$ with their corresponding proofs, it is computationally infeasible to predict $y=f_{sk}(x)$ for any $x \neq x_1, \cdots, x_n$ without the secret key $\sk$. In fact, the distribution of $y$ looks indistinguishable from the uniform distribution on $[0,1]$.

\item 
Unique Provability: For any input $x$, there is exactly one output $y$ that can be verified as the correct computation of $f_{\sk}(x)$.
\end{itemize}

\end{definition}

A balanced scoring function takes in the pseudorandom output generated from the VRF associated with an account (i.e. parametrized by the account's secret key) and the amount of stake in that account, and yields a score. The account with the minimum score is selected as the leader. A balanced scoring function always selects a leader proportional to the account's stake assuming the outputs of the VRFs are truly random. In particular, this implies that splitting one's stake between multiple accounts and/or merging stake with another entity does not impact the probability of being selected as the minimum. This forms the basis for selecting a leader-selection protocol that selects leaders independently in each round proportional to their stake.

% change the distribution of the minimum credential they can generate with the scoring function by partitioning their account into several different accounts with stakes $\alpha_{j1}, \cdots \alpha_{jm}$, where $\sum_{i=1}^m \alpha_{ji} = \alpha_j$.

\begin{definition}[Balanced Scoring Function \cite{FerreiraHWY22}] \label{def:balancedScoringRule}
A \emph{scoring function} $S(\cdot,\cdot)$ takes as input a credential $X_i$ and a quantity of stake $\alpha_i$ and outputs a score $S(X_i, \alpha_i)$. A scoring function is \emph{balanced} if for all $n$ and all player stakes $\alpha_1, \cdots, \alpha_n$, 
\begin{align*}
    \Pr_{X_1, \cdots, X_n \leftarrow U([0,1])} \sq*{\argmin_{i \in [n]}{S(X_i, \alpha_i)} = j} = \frac{\alpha_j}{\sum_{i=1}^n \alpha_i} . 
\end{align*}
\end{definition}

\begin{proposition}[\cite{FerreiraGHHWY24}]
\label{prop:canonical_property}
    Let $S(\cdot,\cdot)$ be any balanced scoring function. Then, for all $n\in \NN$ and $(\alpha_i)_{1\leq i\leq n}$, the random variables $S(X,\sum_{i=1}^{n}\alpha_i)$ and $\min_{1\leq i\leq n}\{S(X_i,\alpha_i)\}$ are identically distributed for $X,X_1,\dots X_n\sim U([0,1])$.
\end{proposition}

\begin{definition}[Cryptographic Self-Selection Protocol (CSSP),~\cite{FerreiraHWY22}]
The Cryptographic Self-Selection Protocol (CSSP) operates as follows:
\begin{enumerate}
    \item 
    Each account $i$, with stake $\alpha_i$, sets up a VRF $f_{\sk_i}(\cdot)$ with a pair of secret key and public key $(\sk_i,\pk_i)$. Participants agree on some Balanced Scoring Function $S(\cdot, \cdot)$.
    \item 
    $Q_r$ denotes the seed used during round $r$. $Q_1$ is a uniformly random draw from $[0,1]$, and $Q_r$ will be determined during round $r-1$ (see below).
    \item 
    In round $r$, each account $i$ computes their credential $\textsc{Cred}^r_{i} = f_{\sk_i}(Q_r)$ using their VRF $f_{\sk_i}$. Each account-holder should broadcast their credential (this is not enforced -- an account-holder may choose not to broadcast, if desired). %Miners who hold multiple accounts should broadcast all their credentials (but this is not enforced).
    \item The leader $\ell_r$ is the account-holder $i$ who broadcasts the credential with the lowest score $S(\textsc{Cred}^r_i,\alpha_i)$.
    \item 
    The seed for the next round, $Q_{r+1}$, is set as the credential of the leader of round $r$, namely $\textsc{Cred}^r_{\ell_r}$.
\end{enumerate}

The actions of a player (who may control multiple accounts) in round $r$ of a CSSP are simply to decide which (if any) of their credentials to broadcast. The payoff to player $i$ is the fraction of rounds in which they are the leader. Formally, if $L_p(r)$ is the indicator variable for whether an account controlled by player $p$ is the leader in round $r$, then the payoff to player $p$ is $\lim\inf_{r \rightarrow \infty} \frac{\sum_{r' \leq r} L_p(r')}{r}$.
\end{definition}

CSSP is a formalization of the leader-selection protocol initiated by Algorand~\cite{ChenM19}. Note that leader-selection is but one aspect of a Proof-of-Stake protocol (the core of the protocol is reaching consensus on the block proposed by the leader). Fortunately, the leader-selection protocols are modular, and can be studied in isolation from the (significantly more complex) consensus algorithms that use them.

\subsection{Strategic Play in CSSP} \label{sec:model:strategy}

Studies of strategic manipulation in consensus protocols first and foremost aim to understand whether one should expect strategic players to choose to be honest. As such, the overwhelming majority of prior work considers a single strategic player against a profile of honest players. ~\cite{FerreiraHWY22} establish that this single strategic player is \emph{not} incentivized to be honest, and we ask whether a strategy that realizes these gains is always detectable. In concurrent and independent work, ~\cite{FerreiraGHHWY24} establish tight bounds on the profitability of manipulations. They do not consider detectability, and therefore the work is orthogonal to ours.

%~\cite{FerreiraHWY22} establishes that this single strategic player is \emph{not} incentivized to be honest,\todo{~\cite{FerreiraGHHWY24}} establish bounds on the gains above honesty, and we ask whether a strategy that realizes these gains is always detectable.

In a CSSP, honest behavior corresponds to broadcasting all credentials in every round. A strategic player may selectively choose which credentials to broadcast in each round. It should initially seem counterintuitive that strategic behavior is profitable -- hiding a credential \emph{in round $r$} certainly cannot help a player \emph{win round $r$}. However, hiding a credential in round $r$ might help a player \emph{win round $r' > r$} by influencing the seed $Q_{r'}$.

%\cite{FerreiraHWY22} gives a detailed analysis on profitable deviations of the strategic player, which we refer to as selfish mining. In this paper we are instead concerned about \textit{detectability} of selfish mining. Here, for the reader's convenience, we will briefly describe the strategy space of the strategic player in CSSP, which was first characterized in \cite{FerreiraHWY22}. Understanding the strategy space is necessary for analyzing the feasibility of a player who can cheat for profit while remaining statistically undetectable. 

\vspace{8pt}

\emph{Strategy Space in CSSP (\cite{FerreiraHWY22}).} 
Consider a CSSP parameterized by $\alpha$, the fraction of stake controlled by the strategic player, and $\beta \in [0, 1]$, the network connectivity strength of the strategic player. The strategic player is called $\beta$-strong if it learns $\beta$ fraction of the credentials broadcast by the honest players before they must broadcast themselves.  Specifically, $\beta = 1$ represents a player that learns all credentials of the honest players before they broadcast (because they are extremely well-connected in the network) and $\beta = 0$ represents a player that learns none of the credentials of the honest players.

\cite{FerreiraHWY22} make refinement of the strategy space of the CSSP game by showing that any strategy of the strategic player is equivalent to a strategy that only broadcasts at most one credential per round, splits their stake into as many accounts as possible, and considers only two honest players $B$ and $C$, the former with $\beta(1-\alpha)$ fraction of the stake, and the latter with $(1-\beta)(1-\alpha)$ fraction of the stake. Their proof shows that for any strategy $s$ in CSSP, you can find another strategy $s'$ in the refined strategy space with the same payoff. Since our focus is on both the profitability and detectability of a strategy, we need to show that such refinement also preserves the detectability of a strategy. Our detection methods (will be introduced in Section \ref{sec:detect}) only assume an access to the broadcast credential with the minimum score (i.e. the leader's credential) in each round. Thus two strategies that induce the same minimum broadcast credential in each round are either both detectable or both undetectable. Since for any undetectable strategy $s$ in CSSP,  there is also an undetectable strategy $s'$ in the refined strategy space,  by having $s'$ broadcast the same credential in each round as the minimum credential that $s$ broadcasts (and broadcasts none if $s$ broadcasts none), it is without loss of generality to consider only refined strategies from now on.  The refined strategy space is described below:

\begin{definition}[Refined CSSP,~\cite{FerreiraHWY22}]
The strategic player first splits their stake in as many account as possible. This set of accounts, denoted as $A$, is then fixed for all rounds. In each round $r$ of CSSP, the strategic player:
\begin{enumerate} 
    \item 
    Is aware of the seed $Q_r$ and the honesty of player $B$ and $C$.
    \item 
    Has access to the credential $\textsc{Cred}^r_B$ of honest player $B$, but not to the credential $\textsc{Cred}^r_C$ of honest player $C$. The player only knows the fact that $\textsc{Cred}^r_C$ is distributed uniformly on $[0,1]$. %$S(\textsc{Cred}^r_C, (1-\beta)(1-\alpha))$ is distributed according to $exp((1-\beta)(1-\alpha))$.
    \item 
    Can compute credentials $\textsc{Cred}^r_i$ and scores $S(\textsc{Cred}^r_i, \alpha_i)$ for accounts $i \in A$, and can compute the score $S(\textsc{Cred}^r_B, \beta(1-\alpha))$.
    \item 
    For each $\ell \in A \cup \{B\}$, can imagine that perhaps $Q_{r+1} = \textsc{Cred}^r_\ell$, and then pre-computes hypothetical credentials $\textsc{Cred}^{r+1}_i$  for each $i \in A$ in case we were to have $\ell_r = \ell$.
    \item 
    Can extend this pre-computation to any round $k$ and sequence of accounts $i_0, \ldots, i_k$ (with $i_0 \in A \cup \{B\}$ and $i_\ell \in A$ for $0< \ell \leq k$), and compute $\textsc{Cred}^{r+k}_{i_k}$ based on the hypothetical possibility that $\ell_{r+\ell} = i_\ell$ for each $l \in \{0, \ldots, k-1\}$.
    \item 
    Selects an account $i^*\in A$ and broadcasts its credential $\textsc{Cred}^r_{i^*}$, or chooses not to broadcast any credential.
\end{enumerate}
\end{definition}

The Refined CSSP is the precise mathematical game we study, for a particular balanced scoring function $S$. In addition, \Cref{claim:canonicalEquivalence} implies that the game induced by CSSP is the same for any balanced scoring function used in the protocol. They further imply that if we have a statistical detection method for strategic behavior under one CSSP protocol using a particular balanced scoring function $S$, we can apply the same detection method to a variant of the protocol using another balanced scoring function $S'$. Therefore, undetectable profitable strategies exist for any leader-selection protocol based on Algorand's cryptographic self-selection if and only if they exist in the Refined CSSP for any particular $S$ of our choosing.

\begin{proposition}[\cite{FerreiraHWY22,FerreiraGHHWY24}] \label{claim:canonicalEquivalence}
The game induced by CSSP with a balanced scoring function is independent of the particular balanced scoring function used. Formally, for two distinct balanced scoring functions $ S, S' $, the games induced by CSSP are identical. Specifically, for all players $ i $, there is a bijective mapping $ f $ from strategies of player $ i $ in the CSSP with $ S $ to strategies of player $ i $ in the CSSP with $ S' $, where all players broadcast the same set of credentials in each round. For all $ i $, the payoff to player $i$ in the CSSP with $ S $ under strategy profile $ s $ is exactly the same as the payoff to $i$ in the CSSP with $ S' $ under strategy profile $ \langle f_i(s_i) \rangle_{i}$.
\end{proposition}

To simplify our analysis, we choose $S(X,\alpha):=-\ln(X)/\alpha$, which induces an exponential distribution with rate $\alpha$, i.e. when $X$ is drawn from $U([0,1])$, $S(X,\alpha)$ is drawn from $\Exp(\alpha)$. The exponential distribution has the nice property that for a set of random variables $X_1, \dots, X_n$ drawn from $\Exp(\alpha_1),\dots, \Exp(\alpha_n)$ respectively, the minimum score $\min_{i\in [n]}X_n$ is distributed according to $\Exp(\sum_{i=1}^{n}\alpha_i)$. This implies that if the total sum of active stakes is 1 and all players are honest, then the minimum score broadcast is distributed according to $\Exp(1)$. Additionally, Lemma \ref{lem:prob_min} and Lemma \ref{lem:min-i-from-exp} imply that the scoring function is a balanced scoring function, and therefore by Proposition \ref{claim:canonicalEquivalence} there is a bijective mapping from strategies of player $i$ in the CSSP game with balanced scoring function $S'$ and strategies of player $i$ in the CSSP game with balanced scoring function $S=-\ln(X)/\alpha$, which achieves the same outcome (i.e., the same leader is selected each round) and the same profit for player $i$. Thus if we can detect any profitable strategy under scoring function $S$, we can use the same method to detect a profitable strategy under scoring function $S'$ since the same strategy is also profitable under $S$. It is therefore without loss of generality to assume $S(X,\alpha)=-\ln(X)/\alpha$ for all the analysis that follows. Appendix \ref{app:sec:exp} contains all the relevant properties of an exponential distribution that we will employ to detect strategic deviation.

% \lindac{Do we even want to assume that adversary can distribution accounts? We may be able to assume that they already did so at the beginning of all rounds. Setting this up each round anyways doesn't sound realistic. }
% We can see that the adversary's action in each round $r$ is confined to 1) distribute their stakes across multiple accounts and 2) choosing which credentials to broadcast. 

% \cite{FerreiraHWY22} points out that the main way for cheating is to manipulate which credential to broadcast. 

\section{Statistical Detection Methods for Strategic Deviations} \label{sec:stat-detec-meth}

Our paper concerns detection of strategic manipulations in CSSP, so we must first clarify what information is available to the onlooker who wishes to distinguish between the case when all participants are honest (but perhaps suffer latency issues), or a strategic player is manipulating the protocol.

We consider the minimal amount of information necessary for an onlooker just to follow the state of the blockchain: the credentials of the leader from each round.\footnote{Note that the detection method of~\cite{BahraniW23} is tailored to Longest-Chain protocols and in particular looks at the distribution of orphans. As there are no orphans in consensus protocols with finality, we need a fundamentally different detection method.} We will show that this information alone suffices to detect any profitable manipulation in case the strategic player has $\beta = 0$. 

Before continuing, we briefly note that the $\beta = 0$ case corresponds to a ``poorly connected'' attacker who cannot learn the broadcasts of other players before deciding their own. This matches the $\gamma = 0$ case when analyzing Proof-of-Work protocols, which is considered standard/canonical. We also note that, if desired, a leader selection protocol could take steps to induce $\beta = 0$ (for example, participants could cryptographically commit to their credential with a large deposit, and then only receive their deposit back upon revealing). Our main result does leverage $\beta = 0$ (and we will highlight where), and it is an interesting technical question to understand the case of $\beta =1$.\footnote{We further explore~\cite{FerreiraHWY22}'s \OLH\ strategy in Section~\ref{sec:example}, which leverages $\beta = 1$ and is detectable.} But, we hope this brief note reminds the reader that $\beta = 0$ is considered the canonical setting. We now proceed with a formal description of the information observed.

\begin{restatable}[Observed distribution]{definition}{defObservedDistribution} \label{def:observedDistribution}  Let an observer pick a uniformly random round $r$ from the set of all rounds $\{1,\dots,R-1\}$. Let $Z, Z_{+1}$ be the random variables denoting the score of the winning credential in consecutive rounds $r$ and $r+1$ respectively. i.e., $Z = S(\cred_{\ell_r}^r, \alpha_{\ell_r})$ and $Z_{+1} = S(\cred_{\ell_{r+1}}^{r+1}, \alpha_{\ell_{r+1}})$. Then $D_Z$ and $D_{Z_{+1}}$ represent the distributions of the winning credentials in round $r$ and $r+1$ when $R \rightarrow \infty$, and we define $F_Z, F_{Z_{+1}}$ to be the cumulative density function (c.d.f.) of $D_Z, D_{Z_{+1}}$ respectively.
\end{restatable}

Let us now briefly discuss a null hypothesis for the observed distribution. One null hypothesis might be that in every round $r$, every player is online and suffers no latency issues (that is, every account-holder $i$ learns of the seed $Q_r$, computes $\cred_i^r$ and broadcasts it within the allotted time-window), and behaves honestly. If this were the case, we would expect the distribution of winning scores to be i.i.d.~from the distribution $S(U([0,1]),1) = \Exp(1)$ (by \Cref{lem:single-account-exp} and \Cref{lem:min-n-exp}), and in particular we would expect $(Z, Z_{+1})$ to be distributed according to $\Exp(1)\times\Exp(1)$. 

This is perhaps too strong a null hypothesis, though -- some participants may go offline for extended periods of time, and there is no reason the rest of the network should a priori be aware of this. Additionally, some participants may be online but suffer latency issues that prevent them from broadcasting their credential in time. We therefore consider a weaker null hypothesis which instead posits that there exists \emph{some} stake $\gamma$ which is online and honest each round, \emph{except the precise value of $\gamma$ is unknown}. Under this null hypothesis, we would expect there to \emph{exist some $\gamma$} for which the distribution of winning scores is i.i.d~from $S(U([0,1]),\gamma) = \Exp(\gamma)$, and therefore we would expect there to exist some $\gamma$ for which $(Z, Z_{+1})$ is distributed according to $\Exp(\gamma)\times \Exp(\gamma)$. For simplicity of notation, we w.l.o.g.~let $1$ denote the `true' online stake, which might be less than the total stake. Therefore, we consider the null hypothesis to also be satisfied when $\gamma > 1$.\footnote{That is, if the strategic player causes it to appear as though a $\gamma > 1$ fraction of the total stake is online and honest, then clearly something is wrong and an onlooker should detect this. But if the true online stake is only $1/3$ of the total stake and the strategic player causes it to appear as though $2 \cdot 1/3$ of the total stake is online and honest, this is plausible to an onlooker who doesn't know the true fraction of online stake.}

Our main result establishes that no profitable strategy for a $\beta = 0$ strategic player induces an observed distribution that passes the null hypothesis. We also consider an even more robust null hypothesis in Section~\ref{sec:detect} where there exists some unknown $\gamma$ for which the fraction of active stake in each round lies in $[(1-\delta)\cdot \gamma,(1+\delta)\cdot \gamma]$, but stick to the simpler null hypothesis first for cleanliness of our main result.

We now elaborate below on two types of statistical tests. Note that in each round, we only have access to the realization, rather than the underlying distribution of the minimum score, so we only have an empirical estimate of $(F_Z,F_{Z_{+1}})$. Still, the number of rounds of history for a Proof-of-Stake-with-Finality blockchain protocol is extremely large. For example, Ethereum produces new blocks every twelve seconds, or 7200 blocks/day. We also remind the reader that all prior analysis on profitability, and the unique prior work on detectability, consider profitability and detectability in steady-state. This is sensible given the intended lifespan of a blockchain and the rate at which blocks are produced.

%Given that the number of total rounds $R$ is sufficiently large (we think of $R$ as tending to infinity), we could instead gain an accurate estimate on the distribution of minimum score in a \textit{uniformly random} round $r$. Therefore, the observer expects $D_Z$ to distribute identically as the true distribution of the leader's scores. Given this, we propose two detection methods based on the expected structure of $D_Z$ and $D_{Z_{+1}}$ when all players honestly follow the protocol.

\vspace{8pt}
\emph{Detection Method 1: Distribution of Minimum Score.} We first focus simply on the distribution of the winning credential across rounds, without looking at correlation of credentials between rounds. Proposition~\ref{prop:DZ-is-exp-gamma} explicitly confirms that under the null hypothesis, $D_Z$ should be $\Exp(\gamma)$ for some $\gamma$.

%We first consider the scenario where the total active stakes are fixed across all rounds. Under this assumption, although the observer does not know the actual amount of active stakes $\lambda$, she should expect $D_Z$ to distribute identically as $\Exp(\gamma)$ for some $\gamma$ that denotes the amount of active stakes. We formally state this as Proposition \ref{prop:DZ-is-exp-gamma}.

\begin{proposition} \label{prop:DZ-is-exp-gamma}
    When the total online stake is constant across rounds and all players honestly broadcast their credentials, there exists a number $\gamma$ such that $D_Z$ is distributed identically to $\Exp(\gamma)$.
\end{proposition}

\begin{proof}
Let $\lambda$ be the actual amount of total online stakes. By Proposition \ref{prop:canonical_property}, $\min_{i}\{S(X_i, \alpha_i)\}$ is distributed identically to $S(X, \sum_{i}\alpha_i)$ when $X, X_i$ are i.i.d. from $U([0, 1])$. Therefore, at each round $r$, the score of the leader $S(\cred_{\ell_r}^r, \alpha_{\ell_r})$ is distributed identically to $S(X, \lambda)$, where $X\sim U([0,1])$ since when all players are honest, $\cred_{i}$ is distributed identically to $U([0,1])$. Thus $S(\cred_{\ell_r}^r, \alpha_{\ell_r})$ is distributed identically to  $\Exp(\lambda)$ by our choice of the scoring function. Thus, the c.d.f. of $D_Z$ is 
% \footnote{\mattnote{I'm not sure if the notation of summing distributions is standard -- it would be OK to use the CDF here and then it's just the literal definition of sum. Or, just define the notation in a footnote.}}
% \begin{alignedequation*}
%     D_Z = \lim_{R\rightarrow \infty}\sum_{r=1}^R \frac{1}{R} D_{S(\cred_{\ell_r}^r, \lambda)} = \lim_{R\rightarrow \infty}\sum_{r=1}^R \frac{1}{R} \Exp(\lambda) = \Exp(\lambda)
% \end{alignedequation*}
\begin{alignedequation*}
    F_Z(z) = \lim_{R\rightarrow \infty}\sum_{r=1}^R \frac{1}{R} F_{S(\cred_{\ell_r}^r, \alpha_{\ell_r})}(z) = \lim_{R\rightarrow \infty}\sum_{r=1}^R \frac{1}{R} (1 - e^{-\lambda z}) = 1 - e^{-\lambda z}
\end{alignedequation*}
which implies that $D_Z$ is distributed identically to $\Exp(\lambda)$. Taking $\gamma=\lambda$ concludes our proof.
\end{proof}

A more robust null hypothesis allows for the possibility that the fraction of online players varies across rounds, but not by much. In this setting, the total amount of online stake lies in some range $[(1 - \delta)\lambda, (1+\delta) \lambda]$ for some small $\delta > 0$ and $\lambda>0$. The exact distribution $D_Z$ is impossible to compute without knowing the actual online stake $\lambda_1, \cdots, \lambda_R$ in each round. Nevertheless, the observer expects that the c.d.f. of $D_Z$ is within a certain range parameterized by $\gamma$ that represents her estimation of $\lambda$.

% \chenghan{A more generalized model allows players to choose not to participate in the game for certain rounds. In this setting, the total amount of stakes are not fixed across rounds, but would instead be in a range $[(1 - \delta)\lambda, (1+\delta) \lambda]$ for some small $\delta > 0$. Because the observer does not know the actual amount of active stakes $\lambda_r$ in round $r$, she is only able to estimate $S(\cred_{\ell_r}^r, \lambda)$ with $Z$, but not $S(\cred_{\ell_r}^r, \lambda_r)$. Nevertheless, since $\lambda_r \in [(1 - \delta)\lambda, (1+\delta) \lambda]$ for all $r$, the observer still expect that $\Exp((1 + \delta)\gamma) \prec D_Z \prec \Exp((1 - \delta)\gamma)$ when all players honestly follow the protocol.}

\begin{proposition}
    When the fraction of online stake lies within a multiplicative $1\pm \delta$ factor across all rounds, and all players honestly broadcast their credentials, there exists a number $\gamma$ such that $\Exp((1 + \delta) \gamma) \preceq D_Z \preceq \Exp((1 - \delta) \gamma)$.\footnote{Here, $D_1 \preceq D_2$ denotes that $D_2$ first-order stochastically dominates $D_1$.}
\end{proposition}

\begin{proof}
    Let $\lambda$ be such that the online stake in each round is within $[(1-\delta) \cdot \lambda, (1+\delta)\cdot \lambda]$. Such $\lambda$ is guaranteed to exist by hypothesis. At each round $r$, since the fraction of online stake in round $r$ is $\lambda_r$, we know that $S(\cred_{\ell_r}^r, \alpha_{\ell_r})$ is distributed according to to $\Exp(\lambda_r)$. Thus, the c.d.f. of $D_Z$ is 
    \begin{alignedequation*}
        F_Z(z) = \lim_{R\rightarrow \infty}\sum_{r=1}^R \frac{1}{R} F_{S(\cred_{\ell_r}^r, \alpha_{\ell_r})}(z) = \lim_{R\rightarrow \infty}\sum_{r=1}^R \frac{1}{R} (1 - e^{-\lambda_r z})
    \end{alignedequation*}
    
    Because $\lambda_r \in [(1 - \delta) \lambda, (1 + \delta) \lambda]$, for all $r$,
    $$\Exp((1 - \delta) \lambda) \preceq \Exp(\lambda_r) \preceq \Exp((1 + \delta) \lambda)$$
    Plugging this in $F_Z(z)$ and substituting $\lambda$ with $\gamma$, we are able to conclude that 
    \begin{alignedequation*}
        \Exp((1 + \delta) \gamma) \preceq D_Z \preceq \Exp((1 - \delta) \gamma)
    \end{alignedequation*}
\end{proof}

We conclude this detection method by reminding the reader that because the observer does not know the actual amount of online stake, $\gamma$, a strategic player could make it appear as though the total online stake is some $\lambda \neq \gamma$ (and we specifically remind the reader that $\gamma > \lambda$ would and should still satisfy our null hypothesis). Our main contribution in this paper is to show that it is impossible to be profitable and preserve the distribution of broadcast scores to be consistent with any fraction of online stake.

% \lindac{Under honest miner behavior, the true distribution $\D_{t}^*$ for each round $t$ should be equal to the exponential distribution with parameter $\lambda_t \in (0, 1]$ equal to the active proportion of stake in round $t$. }

% Hence $\D$ must closely follow $\D^*$, which is a mixture of exponential distribution $\D^* = \sum_{t=1}^T \frac{1}{T} \Exp(\lambda_t)$. We could characterize $\D^*$ under honest mining behavior further based on our knowledge about the stakes in each round. \chenghan{Without loss of generality, let us normalize the expected active proportion of stake to be $1$.} For instance, 
% \begin{itemize}
%     \item 
%     When we know all stake is active in all rounds (i.e. $\forall t \in [T]$, $\lambda_t = 1$), then $\D^* = \Exp(1)$. 
%     \item 
%     When we know that nearly all stake is active in all rounds (i.e., for some small $\gamma$, $\forall t \in [T]$, $\lambda_t \in [1 - \gamma, 1]$), then $\Exp(1 - \gamma) \succ \D^* \succ \Exp(1)$, where $\succ$ represent stochastic dominance relationships. 
%     \chenghan{\item When we know that the active stake in all rounds is around the expected proportion (i.e., for some small $\delta > 0$, $\forall t \in [T]$, $\lambda_t \in [1 - \delta, 1 + \delta]$), then $\Exp(1 - \delta) \succ \D^* \succ \Exp(1 + \delta)$, where $\succ$ represent stochastic dominance relationships. }
% \end{itemize}

\vspace{8pt} \emph{Detection Method 2: Correlation of Consecutive Minimum Scores.} Our second detection method leverages the fact that the credentials of the leader are independent from round to round when all players follow the protocol. In particular, we examine the correlation between the minimum scores in consecutive rounds. Under honest mining behavior, all credentials are drawn i.i.d. from $U([0, 1])$, thus the probability of seeing the score of the leader's credential to be $z_{+1}$ in round $r+1$ should not be changed given the score of the leader's credential $z_r$ in round $r$. Formally, 
\begin{alignedequation*}
    F_{Z, Z_{+1}} (z, z_{+1}) = F_Z(z) \times F_{Z_{+1}}(z_{+1})
\end{alignedequation*}

% \begin{align*}
%     \forall z \in \R^+, \D(Z_{+1} \vert Z = z) = \D(Z_{+1}). 
% \end{align*}
% \lindac{Here maybe need to add a term for statistical error on the order of $O(\frac{1}{T})$.}

% \paragraph{Interpretation.} These definitions provide a basis for detecting strategic manipulation. Deviations from the expected exponential distribution in the minimum score or significant correlations between consecutive minimum scores may indicate strategic behavior by miners, diverging from the presumed honest conduct.

\subsection{Necessary Conditions for Undetectable Strategic Attacks}

In this section, we analyze the effect of the adversary's strategy on $D_Z$ and define the concept of a \emph{statistically undetectable strategy}. The honest players would always broadcast their credentials with the minimum scores (equivalently, broadcast all credentials they have), while the adversary commits to a strategy $\pi$ that does not necessarily broadcast the credential with minimum score.

\begin{definition} \label{def:broadcast-random-vars}
    Let the scoring function $S(\cred, \alpha) = -\ln(\cred)/\alpha$, where $\alpha$ is the stake of the adversary. Pick a round $r$ uniformly at random from the set of all rounds $[R]$, where the total active stake in round $r$ is $\lambda_r$. Let $X_r(\lambda_r), X_{r+1}(\lambda_{r+1})$ be the random variables that denote the minimum score of the honest players' broadcast credentials in round $r$ and $r+1$ respectively; let $Y_r(\pi), Y_{r+1}(\pi)$ be the random variables that denote the score of the adversary's broadcast credential in round $r$ and $r+1$ respectively when they commit to strategy $\pi$. Thus, the score of the leader in round $r$ and $r+1$, could be written as $Z_r(\pi, \lambda_r) = \min\{X_r(\lambda_r), Y_r(\pi)\}$ and $Z_{r+1}(\pi, \lambda_{r+1}) = \min\{X_{r+1}(\lambda_{r+1}), Y_{r+1}(\pi)\}$. 
    % The distribution of random variables $X, X_{+1}, Y(\pi), Y_{+1}(\pi), Z(\pi), Z_{+1}(\pi)$ over a uniformly random round is denoted as $D_{X}, D_{X_{+1}}, D_{Y(\pi)}, D_{Y_{+1}(\pi)}, D_{Z(\pi)}, D_{Z_{+1}(\pi)}$.
    % Then, $D_{Y(\pi)}$ and $D_{Y_{+1}(\pi)}$ are the distribution of the leader's credential in two consecutive rounds, and we define $F_{Y(\pi)}, F_{Y_{+1}(\pi)}$ to be the cumulative density function (c.d.f.) of $D_{Y(\pi)}, D_{Y_{+1}(\pi)}$ respectively.
\end{definition}

We also define the distribution of these random variables with respect to a uniformly random round $r$.
\begin{definition}
    Let $\X$ be a random variable over a uniformly random round. $D_{\X}$ is defined to be the distribution of $\X$, where the corresponding cumulative density function $$F_{\X}(x) = \lim_{R \rightarrow \infty}\Pr_{r \leftarrow U\{1,\ldots, R\}}[\X \leq x].$$
\end{definition}

An observer may choose to examine the distribution of all possible random variables, and even joint distribution of random variables over a random round. For instance, she might examine the score in the previous round, or the joint distribution of the scores in the next $10$ rounds.  
Formally, let $\Z$ denote the set of scores of broadcast credentials that the observer chooses to examine. A strategy $\pi$ is robust to any statistical detection if for any set $\Z$, the joint distributions of seeing all scores in $\Z$ over a random round are identical when the adversary uses strategy $\pi$ with online stake $1$ and when the adversary honestly follows the protocol with online stake $\gamma$. That is, no matter which set of scores the observer chooses to examine, she could not distinguish the distribution when the adversary honestly follows the protocol and there is a $\gamma$ fraction of online stake, or when they use strategy $\pi$ and there is a $1$ fraction of online stake (recall that we w.l.o.g.~let $1$ denote the fraction of online stake for simplicity of notation).

% In this paper, we focus on the two detection methods we proposed in the previous section, since we expect the distribution of minimum scores and the correlation between consecutive minimum scores to follow certain patterns when all players honestly follow the protocol, as we have already discussed. The first detection method that uses the distribution of minimum scores corresponds to the case when the observer chooses to examine the score at each specific round, i.e., $\Z = \{Z_r\}$.
The two detection methods proposed in the previous section give us two necessary conditions for a strategic attack to be undetectable, since we expect the distribution of the minimum scores and the correlation between consecutive minimum scores to follow certain patterns when all players are honest. The first detection method that uses the distribution of minimum scores corresponds to the case when the observer chooses to examine the score at each specific round, i.e., $\Z = \{Z_r\}$.

\begin{definition} \label{def:undetectableMethod1}
   % \footnote{\mattnote{Needs to change to reference that the online stake might change for honest -- I think it's best to just explicitly go for exponential distribution here? Otherwise it'll get notation-heavy (but it's OK to go notation-heavy).} \chenghan{Reolved.}} 
   % A strategy $\pi$ is statistically undetectable to the distribution test if $D_{Z(\pi)}$ is distributed identically to $D_{Z(\pi_{\text{honest}})}$, where $\pi_{\text{honest}}$ is to always broadcast the credential with the minimum score. That is, for any $z_r$,
    % $$\lim_{R \rightarrow \infty} \Pr_{r \leftarrow U(\{1,\ldots, R\})}[Z_r(\pi) \leq z_r] = \lim_{R \rightarrow \infty} \Pr_{r \leftarrow U(\{1,\ldots, R\})}[Z_r(\pi_{\text{honest}}) \leq z_r]$$
    Let $\lambda_r$ be the real participating stake in round $r$. The observer knows that the sequence of participating stakes falls into a certain class $\mathcal{C}_R$ representing a sequence of active stakes (e.g. fluctuation must be within $1\pm \delta$ fraction). The strategy $\pi$ is statistically undetectable to the distribution test if for some $\{\gamma_r\}_{r\in [R]}\in \mathcal{C}_R$ and all $z_r$, 
    $$\lim_{R \rightarrow \infty} \Pr_{r \leftarrow U(\{1,\ldots, R\})}[Z_r(\pi, \lambda_r) \leq z_r] = \lim_{R \rightarrow \infty} \Pr_{r \leftarrow U(\{1,\ldots, R\})}[Z_r(\pi_{\text{honest}}, \gamma_r) \leq z_r].$$
\end{definition}
% \lindac{Maybe use $Z_{r}(\pi_{\text{honest}}, \gamma)$ to illustrate total stake for $Z$? Not sure we need the limit function, since $Z$ is still a random variable here, not realization. }

% Let $\gamma_r^*$ be the real participating stake in round $r$. The observer knows that the sequence of participating stakes fall into a certain class $\mathcal{C}$ (e.g. fluctuation must be within $1\pm \delta$ fraction). The strategy $\pi$ is statistically undetectable if for some $\{\gamma_r\}_{r\in [R]}\in \mathcal{C}$, 
%  $$\Pr_{r \leftarrow U(\{1,\ldots, R\})}[Z_r(\pi, \gamma^*_r) \leq z_r] = \Pr_{r \leftarrow U(\{1,\ldots, R\})}[Z_r(\pi_{\text{honest}}, \gamma_r) \leq z_r].$$

% \begin{definition} \label{def:undetectable} Let $X$ be the random variable of the minimum score of the honest players' credentials and $Y(\pi)$ be the random variable of the adversary's broadcast credentials' score when they use strategy $\pi$. For a uniform random round $r$, the score of the leader's broadcast credential when the adversary plays a \emph{statistically undetectable} strategy is distributed identically to the broadcast credential with minimum score when they follow the intended protocol. That is, $\pi$ is a statistically undetectable strategy if $D_{\min\{X, Y(\pi)\}}$ is distributed identically to $D_{\min\{X, Y(\pi^*)\}}$, where $\pi^*$ is to always broadcast the credential with the minimum score.
% \end{definition}

The second test on correlation between consecutive minimum scores corresponds to the case when the observer chooses to examine the scores in two consecutive rounds, i.e., $\Z = \{Z_r, Z_{r+1}\}$. In order to distinguish with the first test, we leave the constraint on $D_Z$ to Definition \ref{def:undetectableMethod1} and only focus on the correlation between $D_Z$ and $D_{Z_{+1}}$ in Definition \ref{def:undetectableMethod2}.

\begin{restatable}{definition}{defUndetectableMethodTwo} \label{def:undetectableMethod2}
    A strategy $\pi$ is statistically undetectable to the correlation test if $Z(\pi)$ and $Z_{+1}(\pi)$ are independent. That is, for any $z_r$ and $z_{r+1}$,
    $$\lim_{R \rightarrow \infty} \Pr_{r \leftarrow U(\{1,\ldots, R\})}[Z_r(\pi) \leq z_r\ \wedge\ Z_{r+1}(\pi) \leq z_{r+1}] = \lim_{R \rightarrow \infty} \Pr_{r \leftarrow U(\{1,\ldots, R\})}[Z_r(\pi) \leq z_r]\cdot \Pr_{r \leftarrow U[1, R]}[Z_{r+1}(\pi) \leq z_{r+1}]$$
\end{restatable}

\section{A Canonical Example} \label{sec:example}
% \chenghan{Outline:
% \begin{itemize}
%     \item 
%     detection methods: distribution and independence. Discuss why they intuitively make sense 
%     \item 
%     use example to illustrate why our detection methods are useful
% \end{itemize}}

In CSSP, the winning credential of the current round is used as the seed of the next round. This leaves the possibility that an adversary would be strategic in their winning credentials and effectively bias the distribution of seeds. For instance,~\cite{FerreiraHWY22} demonstrate that such protocols are indeed vulnerable to such deviations. In order to acquaint the reader with both the CSSP and statistical detectability, we will show that~\cite{FerreiraHWY22}'s canonical \OLH\ manipulation is statistically detectable using \emph{either} our distribution test or our correlation test. 

Here is some brief intuition for \OLH: because the winning credential is the seed of the next round, the adversary is able to compute credentials for all wallets assuming that a credential in this round is the winning credential. Thus, if the adversary has multiple credentials with low scores to choose from, they could choose to broadcast only the one which maximizes the expected number of rounds won among the current and one-after round. We repeat the formal definition of \OLH\  below:

% An adversary consider their winning probability of both the current round and the next round. 

% Leverage the fact that the winning credential of round $r$ is used as the seed $Q_{r+1}$for the next round. The specific strategy, called \OLH, is defined as follows:

% In their model, the scoring rule \todo{add our notation} is a balanced scoring function $S(x, \alpha_i) := \frac{-\ln(x)}{\alpha_i}$, which inherits desired properties from exponential distributions. 

% To facilitate deeper understanding on $S$, we start by defining the exponential distribution.

% \begin{definition}{(Exponential Distribution)}
% The exponential distribution with rate $\alpha$ is the distribution with cumulative density function (c.d.f.) $F_{\alpha}(x) := 1 - e^{-\alpha x}$, for all $x \geq 0$. We refer to $\Exp(\alpha)$ as one independent sample from the exponential distribution with rate $\alpha$. 
% % For simplicity of notation in later calculations, we will denote by $\text{Exp}(0)$ to be a point-mass at $+\infty$.
% \end{definition}

% In this section, we will show that this 1-\textsc{Lookahead} strategy is detectable using the two methods proposed in section \ref{sec:stat-detec-meth}.

% We will first describe the 1-\textsc{Lookahead} strategy.

\begin{definition}[1-\textsc{Lookahead} strategy] Let the total stake be fixed and normalized to $1$ with the adversary owning an $\alpha$ fraction of the total stake. The goal of the \OLH\ strategy is to maximize the expected number of rounds won among the present and subsequent rounds, and proceeds as follows: 
\begin{enumerate}
    \item Let $r$ be the current round and $A$ be the set of all accounts of the adversary, $B$ be the lone honest account that is broadcast when the adversary decides with total stake $\beta(1-\alpha)$. 
    % Let $W(Q_{r}) = \{i \in A : S(\textsc{Cred}_{i}^{r}, \alpha_{i}) < \min_{j \notin A} S(\textsc{Cred}_{j}^{r}, \alpha_{j})\}$ be the collection of potential winners for the adversary.
    % \item If $|W(Q_{r})| = 0$, broadcast no credentials. Terminate round $r$ and return to Step 1.
    \item Let $W(Q_r) \subseteq A$ denote the accounts $i$ satisfying $S(\cred_i^r,\alpha_i) < S(\cred_B^r, \beta(1-\alpha))$. Observe that $W(Q_r)$ might be empty, and that when $\beta = 0$, $W(Q_r) = A$. 
    % \footnote{\chenghan{I need the parameter $W(Q_r)$ as I need to refer to the round $r$. Should I redefine $W$ to be $W(Q_r)$ to be consistent with later notations, or use $W_r$ as I only need $r$?}}
    \item If $W(Q_r)$ is empty, the adversary cannot win this round, so they move on to the next round and go back to step 1. 
    \item If $W(Q_r)$ is non-empty, for all potential winning accounts $i \in W(Q_r)$ and all potential next-round accounts $j \in A$, compute credential $\textsc{Cred}^{r+1}_{i, j} = f_{\mathsf{sk}_{j}}(\textsc{Cred}_{i}^{r})$, which is the credential of account $j$ in round $r+1$ in the event that account $i$ happens to win round $r$. 
    \item  Let $j(i) = \arg\min_{j \in A} S(\textsc{Cred}^{r+1}_{i,j}, \alpha_j)$ -- this is the account whose credential is most likely to win in round $r+1$ if account $i$ wins round $r$.

    \item For each $i \in W$, define $P^{r+1}_i$ to be the probability that the adversary wins with account $i$ in round $r$ \emph{and} wins with account $j(i)$ in round $r+1$.\footnote{For example, if $\beta = 1$, the probability that the adversary wins with account $i$ in round $r$ is $1$. No matter $\beta$, the probability that the adversary wins with account $j(i)$ in round $r+1$ is just the probability that this credential beats a draw from $\Exp(1-\alpha)$.} That is (below, think of $X^r:=S(\cred_C^r, (1-\beta)(1-\alpha))$):
 \begin{alignedequation*}
       P^{r+1}_i =& \Pr_{X^r \leftarrow \Exp((1-\beta)(1-\alpha))}[S(\cred_{i}^r, \alpha_i) < X^r] \cdot \Pr_{X^{r+1} \leftarrow \Exp(1-\alpha)}[S(\cred_{i, j(i)}^{r+1}, \alpha_{j(i)}) < X^{r+1}].
    \end{alignedequation*}
    
    %Let $j(i) = \arg \min_{j \in A} \Pr_{X^{r+1}\leftarrow \Exp(1-\alpha)}[S(\cred_{i, j}^{r+1}, \alpha_{j}) < X^{r+1}]$, where $ X^{r+1}$ is the random variable of the minimum score of honest players in the subsequent round $r+1$, which is drawn from $\Exp(1 - \alpha)$, and
    %\begin{alignedequation*}
     %   \Pr[S(\cred_{i, j}^{r+1}, \alpha_{j}) < X^{r+1}] =& \Pr[S(\cred_{i}^r, \alpha_i) < X^r] \Pr[S(\cred_{i, j}^{r+1}, \alpha_{j}) < X^{r+1}] + \\
      %  &(1 - \Pr[S(\cred_{i}^r, \alpha_i) < X^r]) \alpha_{j}
    %\end{alignedequation*}
    %$j(i)$ is the credential with minimum score in round $r+1$ given that the winning credential in round $r$ is $\cred_i^r$. \Pr[S(\cred_{i, j(i)}^{r+1}, \alpha_{j(i)}) < X^{r+1}]
    \item Let $i^* = \arg\max_{i \in W} \lt( \Pr_{X^r \leftarrow \Exp((1-\beta)(1-\alpha))}[S(\cred_{i}^r, \alpha_i) < X^r] + P_i^{r+1} \rt)$. $i^*$ is the account that maximizes the expected number of consecutive rounds (among $r, r+1$) that the adversary wins.
    \item Broadcast $\textsc{Cred}^{r}_{i^*}$ at round $r$. If $\cred_{i^*}^r$ is the credential with minimum score in round $r$, broadcast $\textsc{Cred}^{r+1}_{j(i^{*})}$ at round $r + 1$. If $\cred_{i^*}^r$ does not win round $r$ continue.
    % \footnote{\mattnote{Is this right? We should broadcast $j(i^*)$ no matter what? I think we should only do this if we win, and be honest if we lose (or go back to being 1-lookaheady if we lose). }}
    \item Return to Step 1.
\end{enumerate}
\end{definition}

% \chenghan{The above definition seems only consider the winning prob. of round $r+1$, but not of round $r$. I think something should be changed in the above definition.}

% \chenghan{
% \begin{definition}{(Definition 5.1, \cite{FerreiraHWY22} 1-\textsc{Lookahead} strategy)}. The strategy proceeds as follows:

% \begin{enumerate}
%     \item Let $r$ be the current round. Let $W(Q_{r}) = \{i \in A : S(\textsc{Cred}_{i}^{r}, \alpha_{i}) < \min_{j \notin A} S(\textsc{Cred}_{j}^{r}, \alpha_{j})\}$ be the collection of potential winners for the adversary.
%     \item If $|W(Q_{r})| = 0$, broadcast no credentials. Terminate round $r$ and return to Step 1.
%     \item If $|W(Q_{r})| \geq 1$, for each potential winner $i \in W(Q_{r})$, for each account $j \in A$, sample credential $\textsc{Cred}_{i, j}^{r+1} = f_{\mathsf{sk}_{j}}(\textsc{Cred}_{i}^{r})$. Let $j(i) = \arg\min_{j \in A} S(\textsc{Cred}^{r+1}_{i,j}, \alpha_j)$.
%     \item Let $i^{*} = \arg\min_{i \in W(Q_{r})} \min_{j \in A} S(\textsc{Cred}^{r+1}_{i,j}, \alpha_{j})$.
%     \item Broadcast $\textsc{Cred}^{r}_{i^*}$ at round $r$ and $\textsc{Cred}^{r+1}_{j(i^{*})}$ at round $r + 1$.
%     \item Return to Step 1.
% \end{enumerate}
% \end{definition}
% }

While the honest strategy always broadcasts the credential with minimum score and maximizes the probability of winning the current round $r$, \OLH\ instead optimizes the expected number of consecutive rounds won (but only considering the next round -- this is why the strategy is termed \OLH). For example, when $\beta = 1$, and the adversary has multiple accounts that can win this round, they may as well broadcast the credential whose seed gives them the best chance of winning the subsequent round. For $\beta < 1$, the math is trickier, but the strategy always strictly outperforms honesty.

Because the purpose of this section is to gain comfort with the concept of detectability, we focus on the simplest version of \OLH, which is when $\beta =1$ (which corresponds to the most powerful adversary). The arguments in the subsequent subsections proceed roughly as follows:
\begin{itemize}
    \item Section~\ref{sec:example:markov} shows how we might start reasoning about the distribution test (Definition~\ref{def:undetectableMethod1}). In particular, Section~\ref{sec:example:markov}  identifies that we can view the distribution of minimum score $D_{Z_r(\pi_{\OLH})}$ as a mixture of distributions associated with transitions in a two-state Markov Chain, and reasons through what each of these three distributions are. Intuitively, these three distributions are ``what is the minimum broadcast score, conditioned on $r$ being a reset round (i.e. the adversary did not bias $Q_r$ in $r-1$) and the adversary having at least two winning accounts?'', ``what is the minimum broadcast score, conditioned on $r$ being a round where the adversary biased $Q_r$ in $r-1$?'', and ``what is the minimum broadcast score, conditioned on $r$ being a reset round and the adversary has at most one winning account?''. 
    \item Section~\ref{sec:example:notexp} then establishes that no mixture of these distributions can result in an exponential distribution, and therefore $\OLH$ fails the distribution test and is detectable. Intuitively, this follows simply because exponential distributions have a precise rate of tail decay, and the above distributions have no reason to match this precise tail, nor to cancel the differences out.
    \item Section~\ref{sec:example:negative-correlation} considers the correlation test, and establishes that \OLH\ also fails the correlation test. Intuitively, this is because during reset rounds we expect to see a larger than normal winning score (because the adversary may hide coins during a reset round), but during biased rounds we expect to see a lower than normal winning score (because the adversary has biased the seed to make their own score lower than normal). So consecutive rounds are in fact negatively correlated.
\end{itemize}

Note that failing either of the two tests suffice for a strategy to be statistically detectable -- we include both to acquaint the reader with various detection methods (a priori, a strategy might pass one test but fail another).

\subsection{Broadcast Distribution on a Markov Chain} \label{sec:example:markov}

We observe that at each round $r$, the distribution of credentials only depend on the distribution of the seed $Q_r$. This allows us to characterize the CSSP as a stationary Markov chain. For instance, when all players follow the protocol of CSSP and broadcast their credentials that result in the lowest score, the distribution of $Q_r$ is uniformly random from $[0, 1]$ for all $r$. Thus, the Markov chain describing the honest CSSP has only one state that transits to itself with probability $1$. In particular, the game effectively resets when the distribution of $Q_r$ is unbiased. We call such a round to be a ``reset round''. 

% In particular, we define all rounds where the 
% \chenghan{It might be cleaner and more accurate if we do not mention "stopping time" and just call all rounds $\tau$ where $Q_{\tau} \sim U[0, 1]$ as \emph{reset rounds}......}
\begin{restatable}[Reset Round \cite{FerreiraHWY22}]{definition}{defReset}\label{def:reset}
    % A round $r$ is a reset round if the seed of this round is unbiased. i.e., $Q_r \sim U[0, 1]$.
    A round $r$ is a reset round if for all possible strategies $\pi$, the distribution of $\{\Pr[Y_{r'}(\pi) \leq X_{r'} ]\}_{r' \geq r}$ conditioned on $Q_{r-1}$ and all historical information prior to round $r-1$, is identical to the distribution of $\{\Pr[Y_{r'}(\pi) \leq X_{r'} ]\}_{r' \geq r}$ after replacing $Q_{r}$ with a uniformly random draw from $[0, 1]$.
    % \footnote{\mattnote{I don't think this is quite right -- the seed is not drawn from $U([0,1])$, because the seeds are always the minimum credentials. The right definition is more subtle. It's about the distribution of the scores of credentials.} \chenghan{Changed according to Definition of "Stopping Time" in previous paper.}}
\end{restatable}

The \OLH\  strategy, on the other hand, effectively biases the distribution of the seed in favor of the adversary by comparing the best credential for next round. Therefore, the Markov chain of CSSP when the adversary plays \OLH\ is different from the Markov chain of CSSP when every player follows the protocol. 

\begin{restatable}{lemma}{lemOnelookaheadMarkovChain} \label{lem:1-lookahead-markov-chain}
    A CSSP process, with the adversary owning $\alpha$ fraction of the stakes with $\beta = 1$ and using \OLH, is equivalent to the stationary Markov chain\footnote{An illustration of the Markov chain can be found in \Cref{fig:one-lookahead-markov}.} with two states $\Pi = \{C, H\}$, where the transition probability is 
    \begin{alignedequation*}
        \Pr[\Pi_{r + 1} = C | \Pi_{r} = C] =& 1 - \alpha^2 \\
        \Pr[\Pi_{r + 1} = H | \Pi_{r} = C] =& \alpha^2 \\
        \Pr[\Pi_{r + 1} = C | \Pi_{r} = H] =& 1 \\
        \Pr[\Pi_{r + 1} = H | \Pi_{r} = H] =& 0
    \end{alignedequation*}
\end{restatable}
% \chenghan{Change to a transition matrix.}

% Before we prove Lemma \ref{lem:1-lookahead-markov-chain}, we present two useful Lemmas from \cite{FerreiraHWY22} that characterize the distribution of $\min^{(\ell)}_{i} S(\textsc{Cred}^r_i, \alpha_i)$ and $|W(Q_{r})|$.

Standard calculation shows that the stationary distribution of the above Markov chain would be $s_{C} = \frac{1}{1 + \alpha^2}$ and $s_{H} = \frac{\alpha^2}{1 + \alpha^2}$. The following Lemma shows the overall distribution of the leader's credential's score, which is computed by summing the distribution conditioned on each type of transition respectively.

\begin{restatable}{lemma}{lemOverallDistrDZ} \label{lem:overall-distr-DZ} The overall distribution of $D_{Z_r}$ for \OLH\ strategy is
    \begin{alignedequation} \label{eq:main-distr-1-lookahead}
    D_{Z_r} =& \frac{1}{1 + \alpha^2} \lt( \sum_{\ell=1}^{\infty} \Exp_{\ell}(1) \lt[ \sum_{\omega \geq \ell} \frac{\alpha^{\omega}(1 - \alpha)}{\omega}\rt] + (1 - \alpha)\Exp(1) \rt) + \\
    & \frac{1}{1+\alpha^2} \sum_{\omega=2}^{\infty} \alpha^{\omega}(1-\alpha) \Exp(1+(\omega-1)\alpha),
\end{alignedequation}
where $\Exp_{\ell}(1) := \Exp_{\ell - 1}(1) + \Exp(1)$ with $\Exp_0(1) := 0$.
\end{restatable}
%gamma

Equation (\ref{eq:main-distr-1-lookahead}) shows that $D_{Z_r}$ could be viewed as mixture of exponential and Erlang distributions (sum of identical exponential distributions) with different rates. We briefly sketch the argument in the proof of \Cref{lem:overall-distr-DZ}, which is quite technical. Since the score of credential in each account $i$ with stake $\alpha_i$ is distributed identical to an exponential $\Exp(\alpha_i)$ by the properties of exponential distributions (Lemma \ref{lem:min-n-exp} and Lemma \ref{lem:prob_min}), by \Cref{lem:min-i-from-exp}, the $\ell^{th}$ minimum score of credentials that an adversary owns is distributed identical to a Erlang distribution that is sum of $\ell$ identical exponential distributions, denoted as $\Exp_{\ell}$. Since the adversary's action in each round is confined to choosing which credential to broadcast, the adversary, and hence the overall score distribution must be a mixture of $\Exp$ and $\Exp_{\ell}$s with different rates. %to an exponential distribution with different rates.

This key property about $D_{Z_r}$ leads to our main results in this section. In section \ref{sec:example:notexp}, we show that that \OLH\ is statistically detectable under distribution test because the mixture of distributions is not an exponential distribution; 
In section \ref{sec:example:negative-correlation}, we show that \OLH\ is detectable under correlation test because of stochastical dominance relationships between exponential distributions with different rates.

\iffalse
\chenghan{Equation (\ref{eq:main-distr-1-lookahead}) shows that $D_{Z_r}$ could be viewed as exponential distributions with different rates. Intuitively, since the score of credential for each account $i$ is distributed $\exp(\alpha_i)$ by the properties of exponential distributions (Lemma \ref{lem:min-n-exp} and Lemma \ref{lem:prob_min}), the distribution of all strategies available to the adversary are distributed according to an exponential distribution with different rates.
This key property about $D_Z$ leads to our main results in this section. In section \ref{sec:example:notexp}, we show that that \OLH\ is statistically detectable under distribution test because a mixture of exponential distributions is not an exponential distribution; 
In section \ref{sec:example:negative-correlation}, we show that \OLH\ is detectable under correlation test because of stochastical dominance relationships between exponential distributions with different rates.}
\fi

\subsection{Broadcast Distribution of \OLH\ Cannot be an Exponential Distribution} \label{sec:example:notexp}

It is known that the sum of $\ell$ independent exponential variables with mean $1$ each is an Erlang distribution of parameterized by $\ell, 1$. That means, the probability density function (p.d.f.) of $\Exp_{\ell}(z; 1)$ is $\frac{z^{\ell - 1} e^{-z}}{(\ell-1)!}$. Plugging in this and the p.d.f. of exponential distributions, we obtain the p.d.f. of $D_{Z_r}$ to be 

\begin{alignedequation} \label{eq:pdf-1-lookhead}
    f_{Z_r} = &\frac{1}{1 + \alpha^2} \left( \sum_{\ell=1}^{\infty} \frac{z^{\ell - 1} e^{-z}}{(\ell-1)!} \left[ \sum_{\omega \geq l} \frac{\alpha^{\omega} (1 - \alpha)}{\omega} \right] + (1 - \alpha) e^{-x} \right) + \\
    &\frac{1}{1 + \alpha^2} \sum_{\omega = 2}^{\infty} \alpha^{\omega} (1 - \alpha) (1 + (\omega - 1)\alpha) e^{-(1 + (\omega-1)\alpha)}
\end{alignedequation}

If \OLH\ is a statistically undetectable strategy, there exists a parameter $\gamma > 0$ such that $f_Z$ equals to the p.d.f. of $\Exp(\gamma)$. However, the following Lemma shows that this is impossible.

\begin{restatable}{lemma}{reslemonelookaheadnotexp}
\label{lem:1-lookahead-not-exp}
    There is no $\gamma > 0$ such that $D_Z(\pi_{\OLH}) = \Exp(\gamma)$.
\end{restatable}

\begin{proofsketch}
    Assume by contradiction that equation (\ref{eq:pdf-1-lookhead}) is an exponential distribution. i.e., $f_Z = \gamma e^{-\gamma z}$ where $\gamma > 0$ is the amount of active stakes. Rewriting $e^{(1 - \gamma) z}$ and $e^{(1 - \omega) \alpha}$ according to the Taylor expansion of $e^x = \sum_{\ell=1}^{\infty} \frac{1}{(\ell-1)!} x^{\ell-1}$, the coefficient for the $x^{\ell - 1}$ must agree on all $\ell \geq 1$. This means that for all $\ell \geq 1$, 
    $$\gamma^2(1 - \gamma)^{\ell - 1} = \frac{\alpha^\ell (1 - \alpha)}{1 + \alpha^2} \left[ \frac{1}{\ell} + \sum_{\omega = 2}^{\infty} \alpha^{\omega - 1} (1 - \omega)^{\ell - 1}(1 + (\omega - 1)\alpha)^2 \right]$$

    %Consider all $\ell \geq 1$ that is an even number, Since $\omega$ goes from $2$ to $\infty$, $1 - \omega$ is negative for all values of $\omegas$. Hence when $\ell - 1$ is odd, the RHS decreases super exponentially in 
    %For the LHS, $\gamma^2(1 - \gamma)^{\ell - 1} \geq 0$ for all $\gamma > 0$ and $\ell \geq 0$. However, the RHS $ < 0$ when $\alpha \in (0, 1)$, $\ell$ being an even number and $\ell > \frac{1}{\alpha}$.
    We now take the absolute value on both sides, and show that the absolute value on the left hand side and the right hand side does not grow at the same rate with $l$. Therefore, we can conclude that $f_{Z_r}$ cannot be an exponential distribution.
    %For the LHS, $\gamma^2(1 - \gamma)^{\ell - 1} \geq 0$ for all $\gamma > 0$ and $\ell \geq 0$. However, the RHS $ < 0$ when $\alpha \in (0, 1)$, $\ell$ being an even number and $\ell > \frac{1}{\alpha}$. Therefore, we can conclude that $f_Z$ cannot be an exponential distribution.
    
\end{proofsketch}

\subsection{Distribution of Consecutive Two Rounds are Negatively Correlated in \OLH}  \label{sec:example:negative-correlation}

In this section, we apply the correlation test to \OLH\ and show that the distribution of consecutive two rounds are negatively correlated. In a high level, when the adversary successfully hides some credentials in round $r$ and bias $Q_{r+1}$ in round $r$, they have to do so by strategically hiding credentials with minimum scores. The distribution of scores in such a round stochastically dominates the honest distribution. However, the adversary only chooses to hide credentials because they can obtain credentials with lower scores in round $r+1$. Therefore, the distribution of scores in such a round is stochastically dominated by the honest distribution. This establishes a negative correlation between the scores in subsequent rounds. The Lemma states as follows:
\begin{restatable}{lemma}{lemExampleDependence}\label{lem:example-dependence}
    When the adversary uses \OLH\ strategy, the distribution of consecutive two rounds, $D_{Z_r(\pi_{\OLH})}, D_{Z_{r+1}(\pi_{\OLH})}$ are negatively correlated. That is, for any numbers $a, b$,
    $$\Pr_{r \leftarrow U[1, R]}[Z_{r+1} > b | Z_{r }> a] < \Pr_{r \leftarrow U[1, R]}[Z_{r+1} > b]$$
\end{restatable}
We defer the formal proof of \Cref{lem:example-dependence} to \Cref{app:sec:example:negative-correlation}. 

\section{Profitable Strategies are Detectable} \label{sec:detect}
% \section{$\beta = 0$}
In this section, we will show that when the online stake remains constant throughout the protocol, \textit{every} profitable strategy of an adversary with $\beta = 0$ is detectable (and this holds for all $\alpha$).

Let us first highlight a few complexities of detecting profitable manipulations. First, there are certainly undetectable non-profitable manipulations (for example, the adversary could simply never broadcast -- this results in i.i.d.~scores across rounds according to $\Exp(1-\alpha)$, and is indistinguishable from if the adversary were 'non-strategically offline'). Second, note that a strategic adversary can look as far into the (hypothetical) future (assuming they win consecutive rounds) as they like when deciding which accounts to broadcast, and could try to carefully curate them to match a particular distribution. In general, CSSP induces a Markov Decision Process for the adversary, where each state is a countably long list of real numbers. \OLH\ witnesses that the MDP always has a strategy that outperforms honest, and we seek to understand whether any such strategy also satisfies a collection of complex constraints (and more over, there is not a single collection of constraints to satisfy -- the adversary can pick any $\gamma$ and satisfy the undetectability constraints to appear as i.i.d.~$\Exp(\gamma)$).

Given the complexity of the strategy space in CSSP, our proof is surprisingly simple. Firstly, we make use of the following  observation: since the adversary does not know the credentials owned by the honest miner before broadcasting their own,\footnote{This is the key simplifying aspect of our proof that leverages $\beta = 0$.} in order to improve their probability of winning throughout the protocol, the adversary must on average broadcast credentials with smaller scores compared to when they are honest. Simultaneously, as discussed in \Cref{sec:stat-detec-meth}, the observer expects the empirical score distribution to follow an exponential distribution (of undetermined rate $\gamma$). Hence, in order to maintain undetectability, the adversary's credentials must be distributed as an exponential with rate greater than $1$. 

However,  \cite{FerreiraHWY22} shows that unless the adversary controls almost half of the network, the adversary loses to honest participants in a non-trivial fraction of rounds. After such an event, the adversary loses their advantage gained before from strategic manipulation, and must participate as if the protocol has restarted. We call such rounds where adversary regains the perspective of a uniformly random seed ``reset rounds". In a reset round, we show that the adversary must broadcast credentials with scores at least as large as when honest. This leads to a contradiction -- the tail of the ``reset round" credential score distribution is already too fat for the credential score distribution of the adversary to be an exponential of rate greater than $1$. 

Our result also extends to the setting where the active stake fluctuates within $1 \pm \delta$ factor, where we show that any undetectable strategies can only achieve limited profitability bounded by $2 \delta$. 

\subsection{Detectability for Steady Online Stake} \label{subsec:detectConstant}
Throughout \Cref{subsec:detectConstant}, we will assume that the online stake in each round remains constant, and equal to $1$ (by normalization). Among all the online stake, the adversary holds $\alpha$ stake, while the honest participants hold $1 - \alpha$ stake. The outside observer knows that the total online state is steady across rounds, but does not know how much total stake is online.

We first prove that a profitable and undetectable adversary has a score distribution that is strictly dominated by $\Exp(\alpha)$. We will use notations related to the minimum score of broadcast credentials that are formally defined in  \Cref{def:broadcast-random-vars}. 
\begin{theorem} \label{thm:profitStrategyDist}
    When the online stake remains constant throughout the protocol, for any adversary who holds $\alpha$ stake and employs a  profitable and undetectable strategy $\pi$, the adversary's broadcast score $Y_r(\pi)$ from a random round $r$ is distributed identically to $\Exp(\alpha + \epsilon)$ for some  $\epsilon > 0$.
\end{theorem} 

\begin{proof}
Let $X_r(1)$ and $Y_r(\pi)$ be the minimum score of broadcast credential among honest miners and the adversary respectively, at a uniformly random round $r$. Then the overall minimum score at that round is $\min\{X_r(1),Y_r(\pi)\}$. Since the adversary must broadcast before observing the honest miner's credentials in round $r$, $X_r(1)$ is independent of $Y_r(\pi)$. By \Cref{def:undetectableMethod1} and \Cref{prop:DZ-is-exp-gamma}, in order for the adversary's strategic attack to remain undetectable, $\min\{X_r(1),Y_r(\pi)\}$ must distribute according to $\Exp(\gamma)$ for some $\gamma>0$. Since $X_r(1)\sim \Exp(1-\alpha)$ and by independence between $X_r(1)$ and $Y_r(\pi)$, we have that for any $z > 0$, 
\begin{align*}
    &\Pr[\min\{X_r(1),Y_r(\pi)\}\geq z] = e^{-\gamma z}\\
    &\implies \Pr[X_r(1)\geq z]\Pr[Y_r(\pi)\geq z]= e^{-\gamma z}\\
    & \implies e^{-(1-\alpha)z}\Pr[Y_r(\pi)\geq z]= e^{-\gamma z}\\
    &\implies  \Pr[Y_r(\pi)\geq z] = e^{-(\gamma-(1-\alpha))z} = e^{-(\alpha + (\gamma-1))z}.
\end{align*}
Thus $Y_r(\pi) \sim \Exp(\alpha + (\gamma-1))$. The expected fraction of rounds that the adversary wins if they are honest is $\alpha$. Thus to be strictly profitable, the adversary needs to win with fraction $>\alpha$, which requires $\Pr[Y_r(\pi) < 
 X_r(1)]>\alpha$. Since $X_r(1)$ and $Y_r(\pi)$ are exponential random variables with rate $(1-\alpha)$ and $\alpha + (\gamma-1)$, by \Cref{lem:prob_min}, 
\begin{align*}
  \alpha < \Pr[Y_r(\pi) <X_r(1)] = \frac{\alpha + (\gamma-1)}{\gamma}.
\end{align*}
The above equation implies that $\gamma>1$. Thus we conclude  $Y_r(\pi) \sim \Exp(\alpha+(\gamma-1))=\Exp(\alpha+\epsilon)$ for some $\epsilon>0$.

\end{proof}

Now, we show that for a non-trivial fraction of rounds, the adversary's score is drawn from a distribution that dominates $\Exp(\alpha)$. We will need to reason about the adversary's score in reset rounds (defined in \Cref{sec:example} at \Cref{def:reset}), where the distribution of the seed $Q_r$ is unbiased. For the reader's convenience, the formal definition of the reset round is restated here.   

\defReset*

\cite{FerreiraHWY22} shows that the number of reset rounds are non-negligible. 
\begin{lemma}[\cite{FerreiraHWY22}, Theorem 4.1] \label{lem:FHWYFracOfReset}
For $\alpha < \frac{3 - \sqrt{5}}{2} \approx 0.38$, the fraction of rounds that is a reset round is strictly greater than 0.
\end{lemma}

Meanwhile, we show that in a reset round, the adversary's output score distribution  stochastically dominates $\Exp(\alpha)$. 
\begin{claim} \label{claim:resetPointOutput}
    Given that round $r$ is a reset round, adversary's output distribution in round $r$ must (weakly) stochastically dominate $\Exp(\alpha)$. 
\end{claim} 

\begin{proof}
Let $Y_r(\pi)$ be the broadcast minimum coin of the adversary using strategy $\pi$ at a random round $r$. Notice that if $r$ is a reset round, then $Y_r(\pi)$ would be distributed according to $\Exp(\alpha)$ had $\pi$ been an honest strategy. Let $C_1, \dots, C_j$ be the score of credentials of all accounts that the adversary owns at round $r$, where $C_1 \leq C_2\dots \leq C_j$. Then for any $z>0$,
\begin{align*}
    \Pr[Y_r(\pi) < z\mid \text{$r$ is a reset round}]\leq \Pr[C_1<z\mid \text{$r$ is a reset round}] = 1-e^{-\alpha z}
\end{align*}
since $C_1\sim \Exp(\alpha)$ after a reset round. Thus $Y_r(\pi)$'s distribution given that $r$ is a reset round must (weakly) stochastically dominate $\Exp(\alpha)$. 
\end{proof}

Combining \Cref{thm:profitStrategyDist}, \Cref{lem:FHWYFracOfReset} and \Cref{claim:resetPointOutput}, we show a contradiction between above two properties that we have derived about a profitable adversary's score distribution. This shows no profitable adversary strategy is undetectable. 

\begin{theorem} \label{thm:noProfitability}
    When the online stake remains constant throughout the protocol and $\alpha < \frac{3 - \sqrt{5}}{2}$, there is no profitable and statistically undetectable strategy. 
\end{theorem}
\begin{proof}
    Given any adversary strategy $\pi$, let $p_{\pi}$ be the fraction of rounds that is a reset round, by \Cref{lem:FHWYFracOfReset}, $p_{\pi} > 0$. Let $Y_r(\pi)$ be the broadcast minimum coin of the adversary at a random round $r$. Let $Y_{rs}(\pi)$, $Y_{non-rs}(\pi)$ be the broadcast minimum coin of the adversary at a random reset round and at a random non reset round respectively, as defined in \Cref{def:reset}. Then $Y_r(\pi)$ is a mixture of random variable $Y_{rs}(\pi)$ and $Y_{non-rs}(\pi)$ Specifically, $Y = p_{\pi} \cdot Y_{rs}(\pi) + (1-p_{\pi}) \cdot Y_{non-rs}(\pi)$. 
    % Let $F, F_1, F_2$ be c.d.f. of $Y, Y_{rs}, Y_{non-rs}$ respectively, 
    Then for any $z > 0$,
    \begin{align*}
        \Pr[Y_r(\pi) \geq z] = p \cdot \Pr[Y_{rs}(\pi) \geq z] + (1-p) \cdot \Pr[Y_{non-rs}(\pi) \geq z]. 
    \end{align*}
    By \Cref{thm:profitStrategyDist}, for any undetectable strategy, it must be the case that $\Pr[Y_r(\pi) \geq z] \leq e^{-(\alpha + \eps) z}$. However, by \Cref{claim:resetPointOutput} and \Cref{lem:FHWYFracOfReset}, $p > 0$ and $\Pr[Y_{rs}(\pi) \geq z] \geq e^{-\alpha z}$, hence $\Pr[Y_r(\pi) \geq z] \geq p \cdot e^{-\alpha z}$. Since $p$ is not dependent on $z$, it is impossible for $\Pr[Y_r(\pi) \geq z]$ to be both at most $e^{-(\alpha + \eps) z}$ and at least $p \cdot e^{-\alpha z}$ for all $z > 0$.  
\end{proof}

\subsection{Extension to Fluctuation in Online Stakes}

In practice, limited fluctuation in participating stake of the protocol may be expected. In this subsection, we consider the case where the online stake in any round fluctuates within a $1 \pm \delta$ multiplicative factor of the baseline online stake. By normalization, we assume the ground truth baseline online stake is $1$, while the observer anticipates the online stakes to be in $[(1 - \delta)\gamma, (1 + \delta)\gamma]$ for some $\gamma > 0$. We show that any adversary who profits beyond $2 \delta$ probability of winning is detectable. Our key observation is that  \Cref{thm:profitStrategyDist} can be generalized to adversaries who enjoy profits beyond the online stake fluctuation range (as discussed below in \Cref{thm:profitStrategyDistFluctuate}). 

\begin{definition}
An adversary with stake $\alpha$ is $\Delta$-profitable if their probability of being the leader in a random round $r$ is more than $\alpha + \Delta$.
\end{definition}

\begin{theorem} 
\label{thm:profitStrategyDistFluctuate}
    When the online stake fluctuation is known to lie in all rounds within a multiplicative $1\pm \delta$ band of its baseline, for any adversary who holds $\alpha$ stake and employs a  $2\delta$-profitable and undetectable strategy $\pi$, the adversary's broadcast score $Y_r(\pi)$ from a random round $r$ is stochastically dominated by 
    $\Exp(\alpha + \epsilon)$ for some $\epsilon > 0$.
    % has a broadcast distribution $D_{Y_r(\pi)}$ s.t.
    % $$\Exp(\gamma - 1 + \alpha + \delta(\gamma + 1)) \prec D_{Y_r(\pi)} \prec \Exp(\gamma - 1 + \alpha - \delta(\gamma + 1))$$
\end{theorem} 
The proof of \Cref{thm:profitStrategyDistFluctuate} can be found in \Cref{app:sec:detect}. 

\begin{theorem} \label{thm:RobustLimitedProfitability}
    When the online stake fluctuation is known to lie in all rounds within a multiplicative $1\pm \delta$ band of its baseline, no undetectable strategy is $2\delta$-profitable when $\alpha < \frac{3 - \sqrt{5}}{2}$. 
\end{theorem}
\begin{proof}
    The proof is identical to that of \Cref{thm:noProfitability}, where we establish that the adversary cannot produce a strategy whose broadcasts are dominated by $\Exp(\alpha + \epsilon)$. The only difference is that we use \Cref{thm:profitStrategyDistFluctuate} instead of~\Cref{thm:profitStrategyDist} to conclude that this is necessary in order to be undetectable and strictly profitable.
\end{proof}

% Bibliography
\bibliographystyle{alpha}
\bibliography{MasterBib}
\newpage

\appendix
\section{Omitted Proofs in Section \ref{sec:example}} \label{app:sec:example}
\subsection{Omitted Proofs in Section \ref{sec:example:markov}}

Before we proceed to prove Lemma \ref{lem:1-lookahead-markov-chain} and Lemma \ref{lem:overall-distr-DZ}, we will first show several claims on the distribution of $|W(Q)|$ and the distribution of the minimum scores among all players after a reset round. 

    \begin{claim} \label{claim:distr-W-after-reset}
        Conditioned on $r$ being a reset round,
        $$\Pr[|W(Q_r)| = \ell] = \alpha^{\ell} (1 - \alpha)$$
    \end{claim}
    \begin{proof}
        When $r$ is a reset round, $Q_r$ is drawn from $U([0, 1])$. Thus, $\min_{i \in A} S(\cred_i^r, \alpha_i)$ is distributed identical to $\Exp(\alpha)$ and $\min_{j \notin A} S(\cred_j^r, \alpha_j)$ is distributed identical to $\Exp(1 - \alpha)$ by Lemma \ref{lem:min-n-exp}. Applying Lemma \ref{lem:distr-size-of-w}, we have $\Pr[|W(Q_r)| = \ell] = \alpha^{\ell} (1 - \alpha)$.
    \end{proof}

    \begin{claim} \label{claim:distr-min-ell}
        Conditioned on $r$ being a reset round, the $\ell$'th minimum score among all credentials $\min_{i \in A}^{(\ell)} S(\cred_i^r, \alpha_i)$ is identically distributed to $\Exp_{\ell}(1) := \Exp_{\ell - 1}(1) + \Exp(1)$ with $\Exp_0(1) := 0$.
    \end{claim}
    \begin{proof}
        From Lemma \ref{lem:single-account-exp}, the score of each account $S(\cred_i^r, \alpha_i)$ in $i \in A$ is distributed according to $\Exp(\alpha_i)$, and $\min_{i \in A} S(\cred_i^r, \alpha_i)$ is distributed according to $\Exp(\alpha)$. Applying Lemma \ref{lem:min-i-from-exp}, the claim follows.
    \end{proof}

    \begin{claim} \label{claim:distr-min-after-reset}
        Conditioned on $r$ being a reset round, and that the adversary behaves honestly in round $r$, the minimum broadcast score $S(\cred_{l_r}^r, \alpha_{l_r})$ is distributed according to $\Exp(1)$ regardless of whether honest player or adversary wins in round $r$. 
    \end{claim}
        \begin{proof}
    Since $r$ is a reset round, by \Cref{def:reset}, for every account $j$ with stake $\alpha_j$ and VRF $f_{\sk, j}(\cdot)$ that actively participates in round $r$, $\cred_j^{r}$ is distributed identically to $f_{\sk, j}(U[0,1])$. By \Cref{def:vrf}, $f_{\sk, j}(U[0,1])$ is distributed identically to $U[0,1]$. Hence $S(\cred_j^{r}, \alpha_j)$ is distributed identically to $S(U[0,1], \alpha_j)$ which by \Cref{lem:single-account-exp} is equal to $\Exp(\alpha_j)$. Since we assume all participants are honest, all participants broadcast their credentials. By \Cref{lem:min-n-exp}, $\Exp(\alpha_j)$ is distributed identically to the minimum score of $\frac{\alpha_j}{\eps}$ accounts, each with $\eps$ stake. Therefore the minimum broadcast score $S(\cred_{l_r}^r, \alpha_{l_r})$ of round $r$ (even with conditional events) distributes identical as the score $S$ the following process: 
    \begin{enumerate}
        \item 
        We start with $1/\eps$ sub-accounts each with stake $\eps$. Each account $j \in [1/\eps]$ drawn their score $s_j \sim \Exp(\eps)$. 
        \item 
        $S := \min_{j \in [1/\eps]} s_j$. 
        \item 
        $\alpha_j/\eps$ sub-accounts are selected uniformly at random to be associated with account $j$
        % $\alpha/\eps$ accounts are assigned to the adversary, while $\frac{\gamma - \alpha}{\eps}$ accounts are assigned to the honest players. 
    \end{enumerate}
    Clearly, the distribution of $S$ is independent of the account $j$ which wins round $r$. This means the distribution of $S(\cred_{l_r}^r, \alpha_{l_r})$ is independent of whether honest player or adversary wins in round $r$, and is distributed $\sim \min_{j}(\Exp(\alpha_j)) = \Exp(1)$. 
    \end{proof}

We'll now show the proof for Lemma \ref{lem:1-lookahead-markov-chain} and Lemma \ref{lem:overall-distr-DZ}
\lemOnelookaheadMarkovChain*
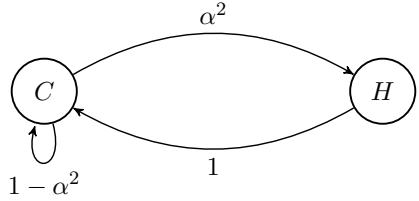
\begin{figure}[h!]
    \begin{center}
    \scalebox{.9}{% \begin{center}
    \begin{tikzpicture}[->, >=stealth', auto, semithick, node distance=5cm]
    \tikzstyle{every state}=[fill=white,draw=black,thick,text=black,scale=1]
    \node[state]    (A)                     {$C$};
    \node[state]    (B)[right of=A]   {$H$};
    % \node    (C)[left of=A]   {$2$};
    % \node[state]    (D)[right of=C]   {$3$};
    \path
    (A) edge[loop below]		node{$1 - \alpha^2$}	    (A)
    (A) edge[bend left,above]	node{$\alpha^2$}	(B)
    (B) edge[bend left,below]	node{$1$}	            (A);
    %\node[above=0.5cm] (A){Patch G};
    %\draw[red] ($(D)+(-1.5,0)$) ellipse (2cm and 3.5cm)node[yshift=3cm]{Patch H};
    \end{tikzpicture}
% \end{center}
}
    \caption{Markov chain of 1-Lookahead strategy. }
    \label{fig:one-lookahead-markov}
    \end{center}
\end{figure}

\begin{proof} [Proof of Lemma \ref{lem:1-lookahead-markov-chain}]
    Let $C$ denote the rounds where the distribution of seed is uniformly random from $[0, 1]$; and $H$ denote all other rounds. At state $C$, the adversary looks at all winning credentials in $W(Q)$ and compare the best option for next round if $|W(Q)| > 1$; at state $H$, the adversary would take advantage of the biased seed distribution and broadcast credential with the lowest score.
    
    Since the adversary is using \OLH\ strategy, the only way they could bias the distribution of $Q_{r}$ is when they have at least two credentials in $W(Q_r)$ and compares which one results in a better credential for round $r+1$. 
    % \chenghan{As shown in Lemma \ref{lem:distr-size-of-w} from \cite{FerreiraHWY22}, which we defer to Appendix \todo{\ref{app:exp}} along with a systematic discussion on properties for exponential distributions, $\Pr[|W(Q_r)| = \ell \big| \Pi_r = C] = \alpha^{\ell} (1 - \alpha)$ since the strategic player's distribution of the minimum score is identical to $\Exp(\alpha)$ after a reset round. }
    Therefore, we could apply Claim \ref{claim:distr-W-after-reset} and derive the probability of staying at $C$ after a reset round:
    % \footnote{\mattnote{The second line below needs justification -- we have not so far cited any result from prior work that justifies this. Presumably, there's a result in that paper we can just cite? (Just to be clear: what I'm claiming is that we haven't cited any fact/lemma that says the probability I have the two smallest coins is $\alpha^2$). Is this Lemma 4.4 which is cited later? It might be good to state Lemma 4.4 explicitly, and explain why it makes sense. The whole purpose of this section is for the reader to understand what's going on, but if Lemma 4.4 is mysterious they might get lost anyway.} \chenghan{Resolved.}} 
    \begin{alignedequation*}
        \Pr[\Pi_{r + 1} = C | \Pi_{r} = C] =& \Pr[|W(Q_r)| < 2 \big| \Pi_r = C] \\
        =& \Pr[|W(Q_r)| = 0 \big| \Pi_r = C] + \Pr[|W(Q_r)| = 1 \big| \Pi_r = C] \\
        =& (1 - \alpha) + \alpha(1 - \alpha) = 1 - \alpha^2
    \end{alignedequation*}
    % \begin{alignedequation*}
    %     \Pr[\Pi_{r + 1} = H | \Pi_{r} = C] =& \Pr[|W(Q_{r})| \geq 2 \big| \Pi_r = C] \\
    %     =& \Pr \lt[ {\min_{i \in A}}^{(2)} S(\textsc{Cred}_i^r, \alpha_i) \leq S(\textsc{Cred}_{j^*}^r, \alpha_{j^*}) \rt] = \alpha^2 
    %     % \Pr[\Pi_{r + 1} = C | \Pi_{r} = C] =& 1 - \alpha^2  
    % \end{alignedequation*}
    , and
    $$\Pr[\Pi_{r + 1} = H | \Pi_{r} = C] = 1 - \Pr[\Pi_{r + 1} = C | \Pi_{r} = C] = \alpha^2  $$
    % When the adversary has only 1 or no credential at state $C$, they are unable to bias the seed of next round. Thus, $\Pr[\Pi_{r+1}=C \big| \Pi_r = C] = 1 - \alpha^2$.
    
    When the process is at status $H$, the distribution of seed is already biased, and the \OLH\ strategy takes advantage of this biased distribution to increase the probability of being the leader (a different strategy \emph{could} decide from this point to still look into the future, but the simple \OLH\ strategy just takes its win and resets). Therefore, they would always broadcast the minimum credential and the distribution of seed of next round becomes unbiased. Therefore, 
    $$\Pr[\Pi_{r + 1} = C \big| \Pi_{r} = H] = 1$$
    
\end{proof}

\lemOverallDistrDZ*
\begin{proof}[Proof of \Cref{lem:overall-distr-DZ}]
Since the total stake is fixed to be $1$ and the strategy used is always \OLH, we let $\lambda_r = 1$ and $\pi$ be $\pi_{\OLH}$ when omitted in $X_r(\lambda_r), Y_r(\pi)$ and $Z_r(\pi, \lambda_r)$ during the whole analysis. Let us use $D_{Z_r | \Pi_r, \Pi_{r+1}}$ to denote the distribution of the score of the winning credential, given that round $r$ is in state $\Pi_r$ and round $r+1$ is in state $\Pi_{r+1}$. As a result, the overall distribution, $D_Z$, could be expressed as follows:
\begin{alignedequation} \label{eq:overall-broadast-sum-of-conditional}
    D_{Z_r} =& \sum_{\Pi_{r}, \Pi_{r+1} \in \{H, C\}} D_{Z_r | \Pi_{r}, \Pi_{r+1}} \cdot \Pr[\Pi_{r + 1} \cap \Pi_r] \\
    =& \sum_{\Pi_{r}, \Pi_{r+1} \in \{H, C\}} D_{Z_r | \Pi_{r}, \Pi_{r+1}} \cdot \Pr[\Pi_{r + 1} | \Pi_r] \cdot s_{\Pi_{r}}\\
\end{alignedequation}

We now proceed to compute each term in equation (\ref{eq:overall-broadast-sum-of-conditional}) separately.

\textbf{Transition distribution when $\mathbf{\Pi_{r} = C}$ and $\mathbf{\Pi_{r+1} = C}$.} Conditioned on transitioning from $C$ to $C$, this means that \emph{both} rounds $r$ and $r+1$ are reset rounds since the adversary has at most 1 winning credential. By Claim \ref{claim:distr-min-after-reset}, the distribution of $Z$ conditioned on transiting from $C$ to $C$ is identical to the distribution when every player honestly follows the protocol:
% \footnote{\mattnote{There's a big lemma missing here which says that in a reset round the distribution of account scores is independent of who owns which accounts (i.e. you can first draw the lowest scores, then afterwards decide how many the adversary had. Probably this is somewhere in the prior paper to cite? But we definitely need to cite it.}}

\begin{alignedequation*}
    D_{Z_r \big| \Pi_{r} = C, \Pi_{r+1} = C} = D_{Z_r}(\pi_{\text{honest}}) = \Exp(1)
\end{alignedequation*}

Therefore, 
\begin{alignedequation} \label{eq:distr-c-c}
    D_{Z \big| \Pi_{r} = C, \Pi_{r+1} = C} \cdot \Pr[\Pi_{r + 1} = C | \Pi_{r} = C] \cdot s_{C} = \frac{1 - \alpha^2}{1 + \alpha^2} \Exp(1)
\end{alignedequation}

\textbf{Transition distribution when $\mathbf{\Pi_{r} = C}$ and $\mathbf{\Pi_{r+1} = H}$.} When $r$ is a reset round, the adversary is only able to bias $Q_{r+1}$ when they have at least $2$ winning credentials. Thus, $\Pi_{r} = C$ and $\Pi_{r+1} = H$ is equivalent to $|W(Q_{r})| \geq 2$ and $\Pi_{r} = C$. Thus,

\begin{alignedequation} \label{eq:distr-z-C-H-interim}
    &D_{Z_r \big| \Pi_{r} = C, \Pi_{r + 1} = H} \cdot \Pr[\Pi_{r} = C \cap \Pi_{r + 1} = H] \\
    =& D_{Z_r \big| |W(Q_{r})| \geq 2, \Pi_{r} = C} \cdot \Pr[|W(Q_{r})| \geq 2 \cap \Pi_r = C] \\
    =& D_{Z_r \big| |W(Q_{r})| \geq 2, \Pi_{r} = C} \cdot \Pr[|W(Q_{r})| \geq 2 \big| \Pi_r = C] \Pr[\Pi_r = C] \\
    =& s_C \cdot \sum_{\omega_{r} = 2}^{\infty} D_{Z_r \big| |W(Q_{r})| = \omega_{r}, \Pi_r = C} \cdot \Pr[|W(Q_{r})| = \omega_{r} \big| \Pi_{r} = C]
\end{alignedequation}

By Claim \ref{claim:distr-W-after-reset}, 
\begin{alignedequation} \label{eq:dist-C-H-card-of-W}
    \Pr[|W(Q_{r})| = \omega_{r} \big| \Pi_{r} = C] = \alpha^{\omega_{r}} (1 - \alpha)
\end{alignedequation}

It remains to compute $D_{Z_r \big| |W(Q_{r})| = \omega_{r}, \Pi_r = C}$. Since the adversary is only able to bias $Q_{r+1}$ when they are the leader in round $r$, 

\begin{alignedequation*}
    D_{Z_r \big| |W(Q_{r})| = \omega_{r}, \Pi_r = C} =& D_{Y_r \big| |W(Q_{r})| = \omega_r, \Pi_r = C} \\
    =& \sum_{\ell = 1}^{\omega_{r}} \Pr \lt[ Y = \text{min}^{(\ell)}_{i} S(\textsc{Cred}^r_i, \alpha_i) \big| |W(Q_{r})| = \omega_r \rt] D_{\min^{(\ell)}_{i} S(\textsc{Cred}^r_i, \alpha_i)} 
\end{alignedequation*}

In 1-\textsc{Lookahead} strategy, $W(Q_{r})$ constitutes the adversary's strategy space in round $r$. Observe that the protocol generates $\textsc{Cred}^{r+1}_{i, j}$ in a random manner. Therefore, when $W(Q_{r}) \neq \emptyset$, the probability that the adversary broadcasts $i$ is the same for all $i \in W(Q_{r})$ and equals $\frac{1}{|W(Q_{r})|}$. Thus, 
$$\Pr \lt[ Y_r = \text{min}^{(\ell)}_{i} S(\textsc{Cred}^r_i, \alpha_i) \big| |W(Q_{r})| = \omega_r \rt] = \frac{1}{\omega_r}$$

% In addition, Lemma \ref{lem:min-i-from-exp} characterizes distribution for the minimum $\ell$'th random variables from a set of i.i.d. exponential distributions with the minimum among all distributed identically to $\Exp(\alpha)$. Applying Lemma \ref{lem:min-i-from-exp} to the all accounts of the strategic player 
% \footnote{\mattnote{Same comment here re: Lemma 4.4. I think we need to give some intuition for where this lemma comes from if we want to use it in a warmup section.}}
By Claim \ref{claim:distr-min-ell}, $\min^{(\ell)}_{i} S(\textsc{Cred}^r_i, \alpha_i)$ is identically distributed to $\Exp_{\ell}(1)$. Therefore, the distribution of scores of broadcast credentials when $\Pi_r = C$ and the adversary has $\omega_r$ winning credentials is 

% Define $\Exp_{\ell}(1) := \underbrace{\Exp(1) + \cdots + \Exp(1)}_{\ell \text{ } \Exp(1)\text{'s}}$ to be the sum of $\ell \text{ } \iid \Exp(1)$. \todo{This is a $\Gamma(\alpha, 1)$ distribution.} From Lemma 4.3 in \cite{FerreiraHWY22}, $\min^{(\ell)}_{i} S(\textsc{Cred}^r_i, \alpha_i)$ is identically distributed to $\Exp_{\ell}(1)$. Thus,

\begin{alignedequation} \label{eq:distr-C-H-distr-given-card-W}
    D_{Z_r \big| |W(Q_{r})| = \omega_{r}, \Pi_r = C} = \sum_{\ell = 1}^{\omega_r} \frac{1}{\omega_r} \Exp_{\ell}(1) 
\end{alignedequation}

Plugging equation (\ref{eq:dist-C-H-card-of-W}) and (\ref{eq:distr-C-H-distr-given-card-W}) back to (\ref{eq:distr-z-C-H-interim}), we have
\begin{alignedequation} \label{eq:distr-z-C-H}
    &D_{Z_r \big| \Pi_{r} = C, \Pi_{r + 1} = H} \cdot \Pr[\Pi_{r} = C \cap \Pi_{r + 1} = H] = \frac{1}{1 + \alpha^2} \sum_{\omega_r = 2}^{\infty} \alpha^{\omega_{r}} (1 - \alpha) \sum_{\ell = 1}^{{\omega_{r}}} \frac{1}{\omega_{r}} \Exp_{\ell}(1)\\
\end{alignedequation}

\textbf{Transition distribution when $\mathbf{\Pi_{r} = H}$ and $\mathbf{\Pi_{r+1} = C}$.} Since the adversary always broadcast the credential with the minimum score at state $H$, the event that $\Pi_{r} = H, \Pi_{r+1} = C$ is equivalent to $\Pi_{r} = H$, which is also equivalent to $|W(Q_{r-1})| \geq 2 \cap \Pi_{r-1} = C$. Thus, 

% \begin{alignedequation*}
%     D_{Z \big| \Pi_{r} = H, \Pi_{r + 1} = C} \cdot \Pr[\Pi_{r} = H \cap \Pi_{r + 1} = C]
% \end{alignedequation*}

\begin{alignedequation*}
    \Pr[\Pi_{r} = H \cap \Pi_{r + 1} = C] = \Pr[\Pi_{r + 1} = C \big| \Pi_{r} = H] \cdot \Pr[\Pi_{r} = H] = s_H = \frac{\alpha^2}{1 + \alpha^2}
\end{alignedequation*}
, and the distribution $D_{Z \big| \Pi_{r} = H, \Pi_{r + 1} = C}$ can be written as
\begin{alignedequation} \label{eq:distr-z-H-C-interim}
    D_{Z_r \big| \Pi_{r} = H, \Pi_{r + 1} = C} =& D_{Z_r \big| \Pi_r = H} 
    = D_{Z_r \big| \Pi_{r-1} = C, |W(Q_{r-1})| \geq 2} 
    \\ =& \frac{\sum_{\omega_{r-1}=2}^{\infty} \Pr[|W(Q_{r-1})| = \omega_{r-1} \big| \Pi_{r-1} = C] \cdot D_{Z_r \big| |W(Q_{r-1})| = \omega_{r-1}, \Pi_{r-1} = C}}{\Pr[|W(Q_{r-1})| \geq 2]}
\end{alignedequation}

In order to derive $D_{Z_r \big| |W(Q_{r-1})| = \omega_{r-1}, \Pi_{r-1} = C}$, we compute the probability that the adversary wins in round $r$. When the adversary has $\omega_{r-1} \geq 2$ winners in round $r-1$ and plays the 1-\textsc{Lookahead} strategy, because they compare $\omega_{r-1}$ credentials in round $r-1$, the score of the broadcast credential in round $r$ is $\min_{i \in W(Q_{r-1})}S(\cred_{j(i)}^{r+1}, \alpha_{j(i)})$ and is exponentially distributed with rate $\alpha \cdot |W(Q_{r-1})|$. Note that, even if the adversary eventually choose to broadcast the credential with minimum score, the distribution of $Z_r$ in round $r$ still deviates from $\Exp(1)$ since the adversary broadcast the credential with minimum score conditioned on "the minimum is the best for next round". Thus, the probability that the adversary wins in round $r$ given $|W(Q_{r-1})| = \omega_{r-1}$ is $\frac{\alpha \omega_{r-1}}{1 + \alpha (\omega_{r-1} - 1)}$, and the broadcast distribution \begin{alignedequation} \label{eq:distr-z-given-card-W(Qr-1)}
    D_{Z_r \big| |W(Q_{r-1})| = \omega_{r-1}, \Pi_{r-1} = C} = \Exp(1 + (\omega_{r-1}-1)\alpha)
\end{alignedequation}

Similar to the analysis in previous cases, the probability of the adversary having $\omega_{r-1}$ winning credentials in a reset round is $\Pr[|W(Q_{r-1})| = \omega_{r-1} \big| \Pi_{r-1} = C] = \alpha^{\omega_{r-1}}(1-\alpha)$, and $\Pr[|W(Q_{r-1})| \geq 2] = \alpha^2$. Plugging these and (\ref{eq:distr-z-given-card-W(Qr-1)}) into (\ref{eq:distr-z-H-C-interim}), we have

\begin{alignedequation} \label{eq:distr-z-from-H-to-C}
    D_{Z_r \big| \Pi_{r} = H, \Pi_{r + 1} = C} =& \sum_{\omega_{r-1}=2}^{\infty} \alpha^{\omega_{r-1}}(1-\alpha) \Exp(1+(\omega_{r-1}-1)\alpha) \Big/ \alpha^2 \\
%    =& \sum_{\omega_{r-1} = 2}^{\infty} \alpha^{\omega_{r-1}} (1 - \alpha) (1 + (\omega_{r-1} - 1) \alpha) e^{-(1 + (\omega_{r-1} - 1) \alpha)} / \alpha^2 
\end{alignedequation}

Summing all the cases together, we finally derive the overall distribution of $D_Z$ for \OLH\ strategy:
% \begin{alignedequation} \label{eq:distr-1-lookahead}
%     D_Z =& D_{Z \big| \Pi_{r} = C, \Pi_{r+1} = C} \cdot \Pr[\Pi_{r + 1} = C | \Pi_{r} = C] \cdot s_{C} + \\
%     & D_{Z \big| \Pi_{r} = C, \Pi_{r + 1} = H} \cdot \Pr[\Pi_{r} = C \cap \Pi_{r + 1} = H] + \\
%     & D_{Z \big| \Pi_{r} = H, \Pi_{r + 1} = C} \cdot \Pr[\Pi_{r} = H \cap \Pi_{r + 1} = C]\\
%     =& \frac{1 - \alpha^2}{1 + \alpha^2} \Exp(1) + \frac{1}{1 + \alpha^2} \sum_{\omega = 2}^{\infty} \alpha^{\omega} (1 - \alpha) \sum_{\ell = 1}^{{\omega}} \frac{1}{\omega} \Exp_{\ell}(1) + \\
%     & \frac{\alpha^2}{1+\alpha^2} \sum_{\omega=2}^{\infty} \alpha^{\omega}(1-\alpha) \Exp(1+(\omega-1)\alpha) \Big/ \alpha^2 \\
%     =& \frac{1}{1 + \alpha^2} \lt( \sum_{\ell=1}^{\infty} \Exp_{\ell}(1) \lt[ \sum_{\omega \geq \ell} \frac{\alpha^{\omega}(1 - \alpha)}{\omega}\rt] + (1 - \alpha)\Exp(1) \rt) + \\
%     & \frac{1}{1+\alpha^2} \sum_{\omega=2}^{\infty} \alpha^{\omega}(1-\alpha) \Exp(1+(\omega-1)\alpha)
% \end{alignedequation}
\begin{alignedequation} \label{eq:distr-1-lookahead}
    D_{Z_r} =& \frac{1}{1 + \alpha^2} \lt( \sum_{\ell=1}^{\infty} \Exp_{\ell}(1) \lt[ \sum_{\omega \geq \ell} \frac{\alpha^{\omega}(1 - \alpha)}{\omega}\rt] + (1 - \alpha)\Exp(1) \rt) + \\
    & \frac{1}{1+\alpha^2} \sum_{\omega=2}^{\infty} \alpha^{\omega}(1-\alpha) \Exp(1+(\omega-1)\alpha)
\end{alignedequation}
\end{proof}
    
\subsection{Ommited Proods in Section \ref{sec:example:notexp}}
\reslemonelookaheadnotexp*
\begin{proof}[proof of \Cref{lem:1-lookahead-not-exp}]
As shown in the main text, the p.d.f. of $D_{Z_r}$ is

\begin{alignedequation} \label{eq:appendix-pdf-1-lookhead}
    f_{Z_r} = &\frac{1}{1 + \alpha^2} \left( \sum_{\ell=1}^{\infty} \frac{z^{\ell - 1} e^{-z}}{(\ell-1)!} \left[ \sum_{\omega \geq l} \frac{\alpha^{\omega} (1 - \alpha)}{\omega} \right] + (1 - \alpha) e^{-z} \right) + \\
    &\frac{1}{1 + \alpha^2} \sum_{\omega = 2}^{\infty} \alpha^{\omega} (1 - \alpha) (1 + (\omega - 1)\alpha) e^{-(1 + (\omega-1)\alpha)z}
\end{alignedequation}

Assume by contradiction that (\ref{eq:appendix-pdf-1-lookhead}) is an exponential distribution. i.e., $f_Z = \gamma e^{-\gamma z}$ where $\gamma > 0$ is the amount of active stakes. 

% \cmt{
% Plug in $x = 0$, we get
% \begin{align*}
%     \lambda =& \frac{1 - \alpha}{1 + \alpha^2}e^{-x} + \frac{1 - \alpha}{1 + \alpha^2} \sum_{i = 2}^{\infty} \alpha^i (1 + (i - 1)\alpha) e^{-(1 + (i-1)\alpha)}\\
%     <& \frac{1 - \alpha}{1 + \alpha^2} + \frac{1 - \alpha}{1 + \alpha^2} \sum_{i = 2}^{\infty} \alpha^i \tag*{$(1 + (i - 1)\alpha) e^{-(1 + (i-1)\alpha)} \leq e^{-1}$ for all $i$} \\
%     =&  \frac{1 - \alpha}{1 + \alpha^2} + \frac{\alpha}{1 + \alpha^2} \\
%     =& 1
% \end{align*}
% Also obvious that $\lambda > 0$ when $\alpha \in (0, 1)$.
% }

Moving terms around, we have
\begin{alignedequation*}
    \gamma e^{(1 - \gamma) z} = &\frac{1}{1 + \alpha^2} \left( \sum_{\ell=1}^{\infty} \frac{z^{\ell - 1}}{(\ell-1)!} \left[ \sum_{\omega \geq l} \frac{\alpha^{\omega} (1 - \alpha)}{\omega} \right] + (1 - \alpha) \right) + \\
    &\frac{1}{1 + \alpha^2} \sum_{\omega = 2}^{\infty} \alpha^\omega (1 - \alpha) (1 + (\omega - 1)\alpha) e^{(1-\omega) \alpha} 
\end{alignedequation*}

Writing the Taylor expansion of $e^{(1 - \gamma) z}$ and $e^{(1 - \omega) \alpha}$, the LHS equals to
\begin{alignedequation} \label{eq:pdf-exponential-taylor-no-consecutive}
    \gamma \sum_{\ell = 1}^{\infty} \frac{(1 - \gamma)^{\ell-1}}{(\ell - 1)!} z^{\ell - 1}
\end{alignedequation}
and the RHS equals to
\begin{alignedequation} \label{eq:pdf-taylor-no-consecutive}
    &\frac{1}{1 + \alpha^2} \left( \sum_{\ell=1}^{\infty} \frac{z^{\ell - 1}}{(\ell-1)!} \left[ \sum_{\omega \geq l} \frac{\alpha^{\omega} (1 - \alpha)}{\omega} \right] + (1 - \alpha) \right) + \\
    &\frac{1}{1 + \alpha^2} \sum_{\omega = 2}^{\infty} \alpha^{\omega} (1 - \alpha) (1 + (\omega - 1)\alpha) \sum_{\ell = 1}^{\infty} \frac{((1 - \omega) \alpha)^{\ell - 1}}{(\ell - 1)!} z^{\ell - 1}
\end{alignedequation}

When (\ref{eq:pdf-exponential-taylor-no-consecutive}) = (\ref{eq:pdf-taylor-no-consecutive}), the coefficient for the $z^{\ell - 1}$ must agree on all $\ell \geq 1$. i.e., the following holds for all $\ell \geq 1$.
\begin{alignedequation*}
    \gamma (1 - \gamma)^{\ell - 1} = \frac{1 - \alpha}{1 + \alpha^2} \left[ \sum_{\omega = \ell}^{\infty} \frac{\alpha^{\omega}}{\omega} + \alpha^{\ell - 1} \sum_{\omega = 2}^{\infty} \alpha^\omega (1 - \omega)^{\ell - 1}(1 + (\omega - 1)\alpha) \right]
\end{alignedequation*}

Taking the difference between $\ell - 1$ and $\ell$ on both sides, LHS becomes
\begin{alignedequation} \label{eq:appendix-lhs-different-ell}
    \gamma (1 - \gamma)^{\ell - 1} - \gamma (1 - \gamma)^{\ell} = \gamma^2 (1 - \gamma)^{\ell - 1}
\end{alignedequation}
and RHS becomes
\begin{align}
    & \frac{1 - \alpha}{1 + \alpha^2} \left[ \sum_{\omega = \ell}^{\infty} \frac{\alpha^{\omega}}{\omega}- \sum_{\omega = \ell + 1}^{\infty} \frac{\alpha^{\omega}}{\omega} + \alpha^{\ell - 1} \sum_{\omega = 2}^{\infty} \alpha^\omega (1 - \omega)^{\ell - 1}(1 + (\omega - 1)\alpha) - \alpha^{\ell} \sum_{\omega = 2}^{\infty} \alpha^\omega (1 - \omega)^{\ell}(1 + (\omega - 1)\alpha)\right] \nonumber \\
    =& \frac{1 - \alpha}{1 + \alpha^2} \left[ \frac{\alpha^{\ell}}{\ell} + \alpha^{\ell - 1} \sum_{\omega = 2}^{\infty} \alpha^\omega (1 - \omega)^{\ell - 1}(1 + (\omega - 1)\alpha)^2 \right] \nonumber \\
    =& \frac{\alpha^\ell (1 - \alpha)}{1 + \alpha^2} \left[ \frac{1}{\ell} + \sum_{\omega = 2}^{\infty} \alpha^{\omega - 1} (1 - \omega)^{\ell - 1}(1 + (\omega - 1)\alpha)^2 \right], \label{eq:pdf-no-consecutive-taylor-difference}
\end{align}
where (\ref{eq:appendix-lhs-different-ell}) needs to equal to (\ref{eq:pdf-no-consecutive-taylor-difference}) for every $\ell \geq 1$.

%However, we observe that for the LHS, $\gamma^2(1 - \gamma)^{\ell - 1} \geq 0$ for all $\gamma > 0$ and $\ell \geq 1$. 
For the RHS, when $\ell$ is an even number and $\ell > \frac{1}{\alpha}$, $(1-\omega)^{\ell - 1}$ is negative for all $\omega$ and thus,
\begin{alignedequation*}
    & \frac{1}{\ell} + \sum_{\omega = 2}^{\infty} \alpha^{\omega - 1} (1 - \omega)^{\ell - 1}(1 + (\omega - 1)\alpha)^2 \\
    \leq& \frac{1}{\ell} + \alpha (-1)^{\ell - 1} (1 + \alpha)^2 + \sum_{\omega = 3}^{\infty} \alpha^{\omega - 1} (1 - \omega)^{\ell - 1}(1 + (\omega - 1)\alpha)^2 \\
    \leq& \frac{1}{\ell} - \alpha + \sum_{\omega = 3}^{\infty} \alpha^{\omega - 1} (1 - \omega)^{\ell - 1}(1 + (\omega - 1)\alpha)^2 
    <  \sum_{\omega = 3}^{\infty} \alpha^{\omega - 1} (1 - \omega)^{\ell - 1}(1 + (\omega - 1)\alpha)^2  < 0. 
\end{alignedequation*}
Therefore for $\ell$ even and $\ell > \frac{1}{\alpha}$, and for any integer $\omega^* \geq 3$
\begin{align*}
    |RHS| &\geq \frac{\alpha^\ell (1 - \alpha)}{1 + \alpha^2} \cdot \sum_{\omega = 3}^{\infty} \alpha^{\omega - 1} |1 - \omega|^{\ell - 1}(1 + (\omega - 1)\alpha)^2\\
    &\geq \frac{\alpha^\ell (1 - \alpha)}{1 + \alpha^2} \cdot \alpha^{\omega^* - 1} |1 - \omega^*|^{\ell - 1}(1 + (\omega^* - 1)\alpha)^2\\
    &=  \frac{\alpha(1 - \alpha)}{1 + \alpha^2} \cdot (1 + (\omega^* - 1)\alpha)^2\cdot \alpha^{\omega^* - 1} \cdot |\alpha(1 - \omega^*)|^{\ell - 1}. 
\end{align*}
Let $g(\alpha, \omega^*) :=  \frac{\alpha(1 - \alpha)}{1 + \alpha^2} \cdot (1 + (\omega^* - 1)\alpha)^2\cdot \alpha^{\omega^* - 1}$, then for any integer $\omega^* \geq 3$, 
\begin{align*}
    |RHS| \geq g(\alpha, \omega^*) \cdot |\alpha(1 - \omega^*)|^{\ell - 1}. 
\end{align*}

Meanwhile, $|LHS| = \gamma^2 \cdot |1 - \gamma|^{l-1}$. If we fix $\omega^*$ and take $l \rightarrow \infty$, LHS = RHS means that $|1- \gamma| \geq \alpha |1 - \omega^*|$ for any $\omega^* \geq 3$. Since $\gamma$ is a fixed parameter, setting $\omega^*$ to be sufficiently large will lead to a contradiction.  We conclude that $f_{Z_r}$(equation \ref{eq:appendix-pdf-1-lookhead}) cannot be an exponential distribution.

%for the LHS, $\gamma^2(1 - \gamma)^{\ell - 1} \geq 0$ for all $\gamma > 0$ and $\ell \geq 1$.

%Combined with the fact that $\frac{\alpha^\ell (1 - \alpha)}{1 + \alpha^2} > 0$ when $\alpha \in (0, 1)$, RHS $ < 0$ when $\ell$ is an even number and $\ell > \frac{1}{\alpha}$. Therefore, there is no $\gamma > 0$ such that (\ref{eq:appendix-lhs-different-ell}) $=$ (\ref{eq:pdf-no-consecutive-taylor-difference}) for all $\ell \geq 1$, which means that there is no $\gamma > 0$ such that $f_Z = \gamma e^{-\gamma z}$. We could then conclude that (\ref{eq:pdf-1-lookhead}) cannot be an exponential distribution.

\end{proof}

\subsection{Omitted Proofs in Section \ref{sec:example:negative-correlation}} \label{app:sec:example:negative-correlation}

\lemExampleDependence*
In order to formally show that $Z_r(\pi_{\OLH})$ and $Z_{r+1}(\pi_{\OLH})$ are negatively correlated, we first compute the joint distribution over scores of broadcast credentials in consecutive two rounds, based on the 3 types of transition in the Markov chain (illustrated in Figure \ref{fig:one-lookahead-markov}).

\begin{alignedequation*}
    D_{CH} :=& D_{Z_r \big| \Pi_{r} = C, \Pi_{r+1} = H} = \sum_{\omega=2}^{\infty} \alpha^{\omega-2} (1 - \alpha) \sum_{\ell = 1}^{\omega} \frac{1}{\omega} \Exp_{\ell}(1) \\
    D_{CC} :=& D_{Z_r \big| \Pi_{r} = C, \Pi_{r+1} = C} = \Exp(1) \\
    D_{H} :=& D_{Z_r \big| \Pi_{r} = H, \Pi_{r+1} = C} = \sum_{\omega = 2}^{\infty} \alpha^{\omega-2} (1 - \alpha) \Exp(1 + (\omega - 1)\alpha) \\
    D_C :=& D_{Z_r \big| \Pi_r = C} = \Pr[\Pi_{r} = C, \Pi_{r+1} = C] \cdot D_{Z_r \big| \Pi_{r} = C, \Pi_{r+1} = C} + \\
    & \qquad \qquad \quad \, \Pr[\Pi_{r} = C, \Pi_{r+1} = H] \cdot D_{Z_r \big| \Pi_{r} = C, \Pi_{r+1} = H}\\
    & \qquad \qquad =\sum_{\omega=1}^{\infty} \alpha^{\omega} (1 - \alpha) \sum_{\ell = 1}^{\omega} \frac{1}{\omega} \Exp_{\ell}(1) + (1 - \alpha) \Exp_1(1)
\end{alignedequation*}

% \chenghan{
% \begin{alignedequation*}
%     D(Success cheat, Honest) \\
%     D(Fail cheat, Cheat) \\
%     D(Honest, Cheat)
% \end{alignedequation*}
% }

The stochastically dominance relationship of $D_{CH}, D_{CC}, D_{H}, D_{C}$ is formally stated in Claim \ref{claim:example-dependence-dominance}.

\begin{claim} \label{claim:example-dependence-dominance}
    For any two distributions $D, D'$ with c.d.f. $F_{D}$ and $F_{D'}$ respectively, we say $D$ (first order) stochastically dominates $D'$ if $F_{D}(x) \leq F_{D'}(x)$ for all $x$, and denoted by $D' \preceq D$. We obtain the following stochastically dominance relationship for $D_H, D_{CC}, D_C, D_{CH}$:
    $$D_{H} \preceq D_{CC} \preceq D_{C} \prec D_{CH}.$$
\end{claim}

\begin{proof}
    By the sum geometric series, we could write $\Exp(1)$ as
    \begin{alignedequation*}
        \Exp(1) =& (1-\alpha) \sum_{\omega = 2}^{\infty} \alpha^{\omega - 2} \Exp(1)
    \end{alignedequation*}
    As $\Exp(1)$ stochastically dominates $\Exp(1 + (\omega - 1)\alpha)$, $D_{CC}$ stochastically dominates $D_{H}$.

    % \todo{Another stochastically dominance}

    Next we derive the stochastically dominance relationship between $D_{CH}$ and $D_{CC}$. By definition of $\Exp_{\ell}(1)$, $\Exp_{\ell}(1) \preceq \Exp_{\ell'}(1)$ when $\ell < \ell'$. Therefore,
    \begin{alignedequation*}
        D_{CH} =& \sum_{\omega=2}^{\infty} \alpha^{\omega-2} (1 - \alpha) \sum_{\ell = 1}^{\omega} \frac{1}{\omega} \Exp_{\ell}(1) \\
        \succ & \sum_{\omega=2}^{\infty} \alpha^{\omega-2} (1 - \alpha) \sum_{\ell = 1}^{\omega} \frac{1}{\omega} \Exp_{1}(1) \\
        =& \Exp(1) = D_{CC}
    \end{alignedequation*}
    Furthermore, $D_C$ is a convex combination of $D_{CH}$ and $D_{CC}$, thus we have $D_{CC} \preceq D_C \preceq D_{CH}$ and the lemma follows.
 
\end{proof}

% \begin{alignedequation*}
%     D_{H} =& \sum_{\omega = 2}^{\infty} \alpha^{\omega} (1 - \alpha) \Exp(1 + (\omega - 1)\alpha) / \alpha^2 \\
%     D_{CC} =& \Exp(1) \\
%     D_{C} =& \sum_{\omega=1}^{\infty} \alpha^{\omega} (1 - \alpha) \sum_{\ell = 1}^{\omega} \frac{1}{\omega} \Exp_{\ell}(1) + (1 - \alpha) \Exp_1(1) \\
%     D_{CH} =& \sum_{\omega=2}^{\infty} \alpha^{\omega} (1 - \alpha) \sum_{\ell = 1}^{\omega} \frac{1}{\omega} \Exp_{\ell}(1)
% \end{alignedequation*}

We are now ready to prove Lemma \ref{lem:example-dependence}.

\begin{proof}
For any constant $a$, since $D_{H} \preceq D_{CC} \preceq D_{CH}$, 
$$\Pr[Z_r > a | \Pi_{r} = C, \Pi_{r+1} = H] > \Pr[Z_r > a | \Pi_{r} = C, \Pi_{r+1} = C] > \Pr[Z_r > a | \Pi_{r} = H].$$ By Bayes' rule, we have

\begin{alignedequation*}
    &\Pr[\Pi_{r} = C, \Pi_{r+1} = H | Z_r > a] \\
    =& \frac{\Pr[Z_r > a | \Pi_{r} = C, \Pi_{r+1} = H] \Pr[\Pi_{r} = C, \Pi_{r+1} = H]}{\Pr[Z_r > a]} 
\end{alignedequation*}

We also know
\begin{alignedequation*}
    \Pr[Z_r > a] =& \Pr[Z_r > a | \Pi_{r} = C, \Pi_{r+1} = H] \Pr[\Pi_{r} = C, \Pi_{r+1} = H] + \\
    &\Pr[Z_r > a | \Pi_{r} = C, \Pi_{r+1} = C] \Pr[\Pi_{r} = C, \Pi_{r+1} = C] + \\
    &\Pr[Z_r > a |\Pi_{r} = H] \Pr[\Pi_{r} = H] \\
    <& \Pr[Z_r > a | \Pi_{r} = C, \Pi_{r+1} = H] \big( \Pr[\Pi_{r} = C, \Pi_{r+1} = H] + \\
    &\Pr[\Pi_{r} = C, \Pi_{r+1} = C] + \Pr[\Pi_{r} = H] \big) \\
    =& \Pr[Z_r > a | \Pi_{r} = C, \Pi_{r+1}=H],
\end{alignedequation*}
where the inequality comes from stochastically dominance property.

Thus,
\begin{alignedequation*}
    \frac{\Pr[Z_r > a | \Pi_{r} = C, \Pi_{r+1} = H] \Pr[\Pi_{r} = C, \Pi_{r+1} = H]}{\Pr[Z_r > a]} > \Pr[\Pi_{r} = C, \Pi_{r+1} = H]
\end{alignedequation*}

The distribution of $Z_{r+1}$ conditioned on $Z_r > a$ is
$$\Pr[\Pi_{r} = C, \Pi_{r+1} = H| Z_r > a] D_{H} + (1 - \Pr[\Pi_{r} = C, \Pi_{r+1} = H | Z_r > a]) D_{C}$$

For any constant $b$, since $D_H \preceq D_C$, $\Pr[Z_{r+1} > b |  \Pi_{r+1} = H] < \Pr[Z_{r+1} > b | \Pi_{r+1} = C]$,
\begin{alignedequation*}
    \Pr[Z_{r+1} > b | Z_{r} > a] =& \Pr[\Pi_{r} = C, \Pi_{r+1} = H | Z_r > a] \Pr[Z_{r+1} > b | \Pi_{r+1} = H] + \\
    &(1 - \Pr[\Pi_{r} = C, \Pi_{r+1}=H | Z_r > a]) \Pr[Z_{r+1} > b | \Pi_{r+1} = C] \\
    =& \Pr[Z_{r+1} > b | \Pi_{r+1} = C] - \Pr[\Pi_{r} = C, \Pi_{r+1}=H | Z_r > a] \big(\Pr[Z_{r+1} > b | \Pi_{r+1} = C] - \\
    & \Pr[Z_{r+1} > b | \Pi_{r+1} = H] \big) \\
    <& \Pr[Z_{r+1} > b | \Pi_{r+1} = C] - \Pr[\Pi_{r} = C, \Pi_{r+1}=H] \big(\Pr[Z_{r+1} > b | \Pi_{r+1} = C] - \\
    & \Pr[Z_{r+1} > b | \Pi_{r+1} = H] \big) \\
    =& \Pr[Z_{r+1} > b]
\end{alignedequation*}

$\Pr[Z_{r+1} > b | Z_{r} > a] < \Pr[Z_{r+1} > b]$ implies that $Z_{r}$ and $Z_{r+1}$ are negatively correlated. 
\end{proof}

\section{Omitted Proofs in Section \ref{sec:detect}} \label{app:sec:detect}

\begin{proof}[Proof of \Cref{thm:profitStrategyDistFluctuate}]
Let $X_r(\lambda_r)$ be the score of the honest miners and $Y_r(\pi)$ be the score of the adversary at a random round $r$. Then the overall minimum broadcast coin at that round is $\min\{X_r(\lambda_r),Y_r(\pi)\}$. We have already observed that $X_r(\lambda_r)$ is independent of $Y_r(\pi)$ in our protocol. 

% Let $D_Z^*$ denote the distribution that the observer anticipates when everyone honestly follows the protocol. Then,

In order to be undetectable, there exits $\gamma > 0$ such that
$$\Exp((1 + \delta)\gamma) \prec D_{\min\{X_r(\lambda_r), Y_r(\pi)\}} \prec \Exp((1 - \delta)\gamma),$$
which means that for all $z > 0$,

$$1 - e^{-(1-\delta) \gamma z} \leq F_{\min\{X_r(\lambda_r), Y_r(\pi)\}}(z) \leq 1 - e^{-(1+\delta) \gamma z}.$$

% Since $X \sim \Exp(1-\alpha)$, we have
% As the adversary always participate in the game, \chenghan{Do we need to reason about this? Since the (canonical) scoring function is decreasing in $\alpha$, the adversary cannot benefit from withholding their stakes. Thus, the adversary always participates with all of their stakes $\alpha$.} the proportion of active stakes hold by the adversary is fixed to be $\alpha$. 

Since the actual distribution $D^*_{Z_r}$ satisfies $\Exp(1 + \delta) \prec D^*_{Z_r} \prec \Exp(1 - \delta)$ for all $r \in [R]$, for any $x > 0$,  

$$e^{-(1 - \alpha + \delta)x} \leq \Pr[X_r(\lambda_r) \geq x] \leq e^{-(1-\alpha-\delta)x}$$

Therefore, a statistically undetectable strategy $\pi$ must satisfy the following where with $\gamma > 0$ and $\gamma_r \in [(1 - \delta)\gamma, (1 + \delta)\gamma]$ for all $r \in [R]$ and $z > 0$:
\begin{alignedequation*}
    &e^{-(1+\delta) \gamma z} \leq \Pr[\min\{X_r(\gamma_r),Y_r(\pi)\}\geq z] \leq  e^{-(1-\delta) \gamma z} \\
    \implies &e^{-(1+\delta) \gamma z} \leq \Pr[X_r(\gamma_r) \geq z]\Pr[Y_r(\pi) \geq z] \leq  e^{-(1-\delta) \gamma z} \\
    \implies& e^{-(\gamma - 1 + \alpha + \delta (\gamma + 1)) z} \leq \Pr[Y_r(\pi) \geq z] \leq e^{-(\gamma - 1 + \alpha - \delta (\gamma + 1)) z}
\end{alignedequation*}
This is equivalent to 
$$\Exp(\gamma - 1 + \alpha + \delta(\gamma + 1)) \prec D_{Y_r(\pi)} \prec \Exp(\gamma - 1 + \alpha - \delta(\gamma + 1))$$

In reality, the honest party has stake $1 - \alpha$ on baseline. Then in any particular round the honest party must have stake at least $1 - \alpha - \delta$. Thus 

\begin{alignedequation*}
    \frac{(1 + \delta) \gamma - (1 - \alpha - \delta)}{(1 + \delta) \gamma} = \frac{\gamma - 1 + \alpha + \delta(\gamma + 1)}{\paren*{\gamma - 1 + \alpha + \delta(\gamma + 1)} + (1 - \alpha - \delta)} \geq \Pr[Y_r(\pi) < X_r(\lambda_r)].
\end{alignedequation*}
Since the adversary is $2 \delta$-profitable, and that the total stake in the game is at most $1 + \delta$, 
\begin{alignedequation*}
    \Pr[Y_r(\pi) < X_r(\lambda_r)] >& \frac{\alpha + 2\delta}{1 + \delta} . 
\end{alignedequation*}
We conclude that 
\begin{alignedequation*}
    &&\frac{(1 + \delta) \gamma - (1 - \alpha - \delta)}{(1 + \delta) \gamma}&> \frac{\alpha + 2\delta}{1 + \delta}\\
    \Rightarrow&&  1 - \frac{1 - \alpha - \delta}{(1 + \delta)\gamma} &> \frac{\alpha + 2\delta}{1 + \delta}\\
    \Rightarrow&& \gamma &> 1.
\end{alignedequation*}
Therefore $Y_r(\pi) \sim \Exp(\alpha + (\gamma - 1)) = \Exp(\alpha + \epsilon)$ for some $\epsilon > 0$.  
\end{proof}

\section{Relavent Properties for Exponential Distributions}
\label{app:sec:exp}
\begin{lemma}[\cite{FerreiraHWY22}, Lemma 2.1.] \label{lem:single-account-exp}
    Let $S(x, \alpha_i) := \frac{-\ln(x)}{\alpha_i}$. When $x$ is drawn uniformly from $[0, 1]$, $S(x, \alpha_i)$ is identically distributed to $\Exp(\alpha_i)$.
\end{lemma}

\begin{lemma}[\cite{FerreiraHWY22}, Lemma A.1.] \label{lem:min-n-exp}
    Let $X_1, \cdots, X_n$ be independent random variables where $X_i$ is drawn from $\Exp(\alpha_i)$ for some $\alpha_i>0$. Then $\min_{i \in [n]}\{X_n\}$ is identically distributed to $\Exp \lt( \sum_{i=1}^n \alpha_i \rt)$.
\end{lemma}

\begin{lemma}[\cite{FerreiraHWY22}, Lemma A.2.] \label{lem:prob_min}
    Let $X_1, X_2$ be two independent random variables drawn from $\Exp(\alpha_1), \Exp(\alpha_2)$ respectively, where $\alpha_1, \alpha_2>0$ . Then $\Pr[X_1<X_2]=\frac{\alpha_1}{\alpha_1+\alpha_2}$.
\end{lemma}

\begin{lemma}[\cite{FerreiraHWY22}, Lemma 4.3.] \label{lem:min-i-from-exp}
    Let $X_1, X_2, \ldots$ be i.i.d. copies of an exponentially distributed random variable such that $\min_{n \in \mathbb{N}} X_n$ is exponentially distributed with rate $\alpha$. Then, for all $i \in \mathbb{N}$, the random variable $Y_i = \min_{n \in \mathbb{N}}^{(i)} X_n$ is identically distributed to $Z_i = Z_{i-1} + \text{Exp}(\alpha)$ where $Z_0 := 0$.
\end{lemma}

\begin{lemma}[\cite{FerreiraHWY22}, Lemma 4.4.] \label{lem:distr-size-of-w}
    Let $Y_1, Y_2, \cdots$ be i.i.d. copies of an exponentially random variable such that $\min_{n \in \mathbb{N}} Y_n$ is exponentially distributed with rate $\alpha$. Let $X$ be exponentially distributed with rate $1 - \alpha$. Let $W = \{i \in \mathbb{N}: Y_i < X\}$. Then $\Pr[|W| = \ell] = \alpha^{\ell}(1 - \alpha)$.
\end{lemma}

\end{document}